\documentclass{sig-alternate-2013}
\usepackage{amssymb}
\usepackage{amsmath}
\usepackage[dvipsnames]{xcolor}
\usepackage{kantlipsum}
\usepackage{framed}
\usepackage{url}
\usepackage{epsfig}
\usepackage{graphicx}
\usepackage
[pdfstartview=FitH,colorlinks,linkcolor=red,citecolor=blue]{hyperref} 
\usepackage{breakurl}
\usepackage{flushend}

\usepackage{amsthm}

\usepackage{pgfplots}
\usetikzlibrary{patterns}

\newcommand{\ignore}[1]{}

\long\def\abbr#1#2{#2}     % long version

\DeclareRobustCommand*{\reflemmanumth}{\ref*{lemma:numth}}

\newfont{\mycrnotice}{ptmr8t at 7pt}
\newfont{\myconfname}{ptmri8t at 7pt}
\let\crnotice\mycrnotice%
\let\confname\myconfname%

\permission{Permission to make digital or hard copies of all or part of this
work for personal or classroom use is granted without fee provided that
copies are not made or distributed for profit or commercial advantage and
that copies bear this notice and the full citation on the first page.
Copyrights for components of this work owned by others than the author(s)
must be honored. Abstracting with credit is permitted. To copy otherwise,
or republish, to post on servers or to redistribute to lists, requires
prior specific permission and/or a fee. Request permissions from
permissions@acm.org.} 
\conferenceinfo{CCS'13,}{November 4--8, 2013, Berlin, Germany. \\
  {\mycrnotice{Copyright is held by the owner/author(s). Publication rights
    licensed to ACM.}
  }
} 
\copyrightetc{ACM \the\acmcopyr}
\crdata{978-1-4503-2477-9/13/11\ ...\$15.00.\\
  http://dx.doi.org/10.1145/2508859.2516680}

\makeatletter
   \abbr{}{\def\@copyrightspace{\relax}}
\makeatother

\newtheorem{theorem}{Theorem}

\newtheorem{lemma}[theorem]{Lemma}
\newtheorem*{lemma*}{Lemma}

\newtheorem{fact}[theorem]{Fact}

\renewcommand{\paragraph}[1]{\vspace{0.8\baselineskip plus 0.8\baselineskip minus 0.5\baselineskip}%
  \noindent\textbf{#1.} }

\long\def\xxx#1{}

\def\polylog{{\mathop{\mathrm{polylog }}\nolimits}}

\makeatletter

\begin{document}

\title{Ensuring High-Quality Randomness\\
  in Cryptographic Key Generation}
\abbr{}{
\subtitle{(Extended version)\titlenote{%
This is an extended and corrected version of a paper which appeared in the
proceedings of the 2013 ACM Conference on Computer and Communications
Security (CCS). 
The proceedings version contained an error in the DSA protocol and
accompanying security proof.
This version corrects that error and contains evaluation results
for the revised protocol.\\
\-\ \ To ensure that the our RSA and DSA evaluations reflect similar network
conditions, we reran the network experiments whose results we summarize in
Table~\ref{tab:timing} and Figure~\ref{fig:keysize}.
The cross-country round-trip network latency has increased to $100$~ms
from $80$~ms since we ran the experiments for the proceedings version 
of the paper.\\
\-\ \ This version also contains the full proof of security for the
RSA protocol.}}
}

% CCS reviewing is double-blind, so we should
% comment out the author names
\numberofauthors{4}
\author{
\alignauthor Henry Corrigan-Gibbs\abbr{\titlenote{
Work conducted while author was a staff member
at Yale University.}}{}\\
  \affaddr{Stanford University}\\
  \email{henrycg@stanford.edu}
\alignauthor Wendy Mu\\
  \affaddr{Stanford University}\\
  \email{wmu@cs.stanford.edu}
\and
\alignauthor Dan Boneh\\
  \affaddr{Stanford University}\\
  \email{dabo@cs.stanford.edu}
\alignauthor Bryan Ford\\
  \affaddr{Yale University}\\
  \email{bryan.ford@yale.edu}
}

\maketitle

\begin{abstract}
The security of any cryptosystem
relies on the secrecy of the system's secret keys.
Yet, recent experimental work demonstrates
that tens of thousands of devices on the Internet
use RSA and DSA secrets drawn from a 
small pool of candidate values.
As a result, an adversary can derive the 
device's secret keys without breaking 
the underlying cryptosystem.
We introduce a new threat model, 
under which there is a {\em systemic} solution 
to such randomness flaws.
In our model, when a device generates a cryptographic key, 
it incorporates some random values from an {\em entropy authority} into
its cryptographic secrets and then {\em proves} to the authority,
using zero-knowledge-proof techniques, that it
performed this operation correctly.
By presenting an entropy-authority-signed 
public-key certificate to a third party 
(like a certificate authority or SSH client),
the device can demonstrate that its 
public key incorporates randomness from the authority
and is therefore drawn from a large pool of candidate values.
Where possible, our protocol protects against eavesdroppers,
entropy authority misbehavior, and devices attempting
to discredit the entropy authority.
To demonstrate the practicality of our protocol,
we have implemented and evaluated its performance on a 
commodity wireless home router.
When running on a home router, our protocol 
incurs a $2.1\times$ slowdown over
conventional RSA key generation and it 
incurs a $4.4\times$ slowdown over conventional
EC-DSA key generation.
\end{abstract}

\abbr{ % Only put this ACM stuff in the camera-ready version

}{}

\section{Introduction}
A good source of randomness is crucial for 
a number of cryptographic operations.
Public-key encryption schemes use randomness to
achieve chosen-plaintext security,
key-exchange algorithms use randomness to 
establish secret session keys, and
commitment schemes use randomness to 
hide the committed value.
The security of these schemes relies
on the unpredictability of the random input values, 
so when the ``random'' inputs are not really 
random, dire security failures 
result~\cite{%
%cve-2000-0357,
cve-2001-0950,
cve-2001-1141,
cve-2001-1467,
cve-2003-1376,
cve-2005-3087,
cve-2006-1378,
cve-2006-1833,
cve-2007-2453,
cve-2008-0141,
cve-2008-0166,
cve-2008-2108,
cve-2008-5162,
%cve-2009-0255,
cve-2009-3238,
cve-2009-3278,
cve-2011-3599,
goldberg96randomness,heninger12mining,lenstra12ron,yilek09private}.

Although the dangers of weak randomness have been part of the
computer security folklore for years~\cite{goldberg96randomness},
entropy failures are still commonplace. 
In 2008, a single mistaken patch caused the OpenSSL distribution
in all Debian-based operating systems 
to use only the process ID (plus a few other easy-to-guess values) 
as the seed for its pseudo-random number generator.
This bug caused affected machines to select
a 1024-bit RSA modulus from a pool of fewer than one million
values, rather than the 
near-$2^{1000}$ possible values~\cite{yilek09private}.
By replaying the key generation process using each of the
one million possible PRNG seeds, an adversary could
recover the secret key corresponding to one of these 
weak public keys in a matter or hours or days.

Recent surveys~\cite{heninger12mining,lenstra12ron}
of SSH and TLS public keys on the Internet demonstrate
that hardware devices with poorly seeded random number generators
have led to a proliferation of weak cryptographic keys.
During the drafting of this paper, NetBSD maintainers discovered
a bug caused by a ``misplaced parenthesis'' that could have caused
NetBSD machines to generate 
cryptographic keys incorporating 
as few as 32 or 64 bits of entropy,
instead of the 100+ expected bits~\cite{netbsd13security}.
Even more recently, a PRNG initialization bug in 
the Android operating system could have caused
applications using the system to generate weak
cryptographic keys~\cite{klyubin13some}.

Randomness failures continue to haunt
cryptographic software for a number of reasons:
the randomness ``stack'' in a modern operating system~\cite{ristenpart10good}
is large and complex,
there is no simple way to test whether a
random number generator is really generating random
numbers, and (at least in the context of cryptographic keys)
there has never been
a {\em systemic solution} to the randomness problem.
The response to entropy failures
has traditionally been {\em ad hoc}: each device manufacturer or
software vendor patches RNG-related bugs 
in its own implementation (once discovered),
without deploying techniques 
to prevent similar failures in the future.
The quantity and severity of randomness failures
suggests that this ``fix the implementation'' 
approach is grossly insufficient. 

We offer the first systemic solution to the entropy problem
in cryptographic key generation for public-key cryptosystems.
In our protocol, a device generating a cryptographic
keypair fetches random values from an {\em entropy authority}
and incorporates these values into its cryptographic secrets.
The device can later {\em prove} to third parties 
(e.g., a certificate authority or an SSH client) 
that the device's secrets incorporate the authority's 
random values, thus guaranteeing
that the device's cryptographic keys are selected
from a large enough pool of candidate values.
Unlike certificate authorities in today's Internet,
our entropy authorities are {\em not} trusted third 
parties: if the device has a strong entropy source,
a malicious entropy authority learns no useful
information about the device's secret key.
We present versions of our protocol for 
RSA and DSA key generation and we 
offer proofs of security for each.

A subtlety of our solution is the threat model: 
under a traditional ``global passive adversary''
model, the adversary can completely simulate
the view of a device that has a very weak entropy source.
Thus, under the global passive adversary model,
a device with a weak entropy source has no hope of
generating strong keys.
We propose an alternate threat model, 
in which the adversary can observe all communication
{\em except for} one initial communication session between
the device and the entropy authority.
Under this more limited adversary model, which is
realistic in many deployment scenarios, we
can take advantage of an entropy authority to 
ensure the randomness of cryptographic keys.

The key generation protocols we present
are useful both for devices with strong and weak entropy sources.
In particular, if the device has a strong entropy source
(the device can repeatedly sample from the uniform distribution over a large set of values), 
running the protocol {\em never weakens} the device's cryptographic keys.
In contrast, if the device has a weak or biased entropy source, running
the protocol can {\em dramatically strengthen} 
the device's keys by ensuring that its keys 
incorporate sufficient randomness.   The device need not know 
whether it has a strong or weak entropy source:
the same protocol is used in both cases. 

A recent survey of public keys~\cite{heninger12mining} 
suggests that embedded devices are responsible for generating 
the majority of weak cryptographic keys on the Internet.
To demonstrate that our protocols are practical even 
on this type of computationally limited network device, we have 
evaluated the protocols on a \$70 Linksys home router
running the dd-wrt~\cite{dd-wrt} operating system. 
Our RSA key generation protocol incurs a 
$2.1\times$ slowdown
on the Linksys router when generating a 2048-bit key, 
and our RSA and DSA protocols incur no more than 
$2$ seconds of slowdown on a laptop and a workstation.
The DSA version of our protocol 
is compatible with both the elliptic-curve 
and finite-fields DSA variants.
Our protocols generate standard RSA and DSA keys
which are, for a given bit-length, as secure as
their conventionally generated counterparts.

In prior work, Juels and Guajardo~\cite{juels02verifiable} 
present a protocol in which a possibly malicious
device generates an RSA key in cooperation with a certificate authority.
Their protocol prevents a device
from generating an {\em ill-formed keypair}
(e.g., an RSA modulus that is the product of more than two primes).
We consider a different threat model.
We ensure that a device samples its keys 
from a distribution with high min-entropy, but we do not prevent
the device from generating malformed keys.
Under this new threat model, 
we achieve roughly a $25\times$ performance improvement
over the protocol of Juels and Guajardo 
(as measured by the number of 
modular exponentiations that the device must compute).
Section~\ref{sec:rel} compares the two protocols and
discusses other related work.

After introducing our threat model in Section~\ref{sec:model},
we describe our key generation protocols in Section~\ref{sec:proto}
and present security proofs in Section~\ref{sec:sec}.
Section~\ref{sec:eval} summarizes our evaluation results
and Section~\ref{sec:impl} discusses issues related to 
integrating our protocols with existing systems.
 
\subsection{Why Other Solutions Are Insufficient}

Before describing our protocol in detail, we discuss
a few other possible, but unsatisfactory, ways 
to prevent networked devices 
from using weak cryptographic keys.\footnote{%
By {\em weak keys} we mean keys sampled from
a distribution with much less min-entropy than
the user expects. 
For example, a 224-bit EC-DSA key sampled
from a distribution with only 20 bits of min-entropy
is weak.
}

\paragraph{Possible Solution \#1: Fix the implementation}
One possible solution to the weak key problem is
to simply make sure that cryptography libraries properly 
incorporate random values into the cryptographic 
secrets that they produce.
Unfortunately, bugs and bad implementations are a fact of
life in the world of software, and 
the subtleties of random number generation 
make randomness bugs particularly common. 
Implementations that seed their random number generators
with public or guessable values 
(e.g., time, process ID, or MAC 
 address)~\cite{cve-2001-1141,
cve-2008-0141,
cve-2008-0166,
cve-2008-5162,
cve-2008-5162%
%cve-2009-0255
},
implementations that use weak random number 
generators~\cite{cve-2003-1376,cve-2005-3087,cve-2006-1378,
  cve-2009-3278,
  cve-2011-3599},
and implementations without a good source
of environmental entropy~\cite{heninger12mining}
are all vulnerable.

The complexity of generating cryptographically strong random numbers,
the overwhelming number of randomness failures in 
deployed software, and the difficulty of detecting
these failures during testing all indicate
that ``fix the implementation'' is an insufficient
solution to the weak key problem.
Given that some implementations will be buggy, there should
be a way to {\em assure} clients that their TLS and SSH servers are using
strong keys, even if the client suspects that the
servers do not have access to a good source of random values.

\paragraph{Possible Solution \#2: Simple entropy server}
A second possible solution would be to have devices fetch
some random values from an ``entropy sever'' and incorporate
these values (along with some random values that the device picks)
into the device's cryptographic secrets.
As long as the adversary cannot observe the device's communication
with the server, the server would provide an 
effective source of environmental entropy.

One problem with this approach comes in attributing
blame for failures.
If a device using an entropy server produces
weak keys, the device might {\em blame the entropy server} 
for providing it with weak random values.
In turn, the entropy server could claim that it provided the device
with strong random values but that the device failed to 
incorporate them into the device's cryptographic secrets.
Without some additional protocol, a third party will not be
able to definitively attribute the randomness failure to either
the device or the entropy server.

\paragraph{Possible Solution \#3: Key database}
A third possible technique to 
prevent devices from using weak keys would be to deploy
a ``key database'' that contains a copy of every public
key on the Internet.
A non-profit organization could run this database, much
as the Electronic Frontier Foundation maintains the
SSL Observatory~\cite{eff-cert},
a static database of public keys on the Internet.

Whenever a device with a 
potentially weak entropy source generates a 
new keypair, the device would send its new public 
key to the key database.
If the database already contains that key (or if the
database contains an RSA modulus that shares a factor
with the new key), the device would generate a fresh key
and submit it to the database.
The device would continue this generate-and-submit 
process until finding a unique key.
At the end of the process, the device would be guaranteed
to have a key that is unique, at least amongst
the set of keys in the database. 

Unfortunately, this proposed solution would {\em obscure} 
the entropy problem without fixing it.
An attacker could replay the entire key generation process
using the known initial state of a device with 
a weak entropy source to learn
the secret keys of that device.
By creating a centralized database of (possibly weak) keys, 
such a solution would make it easier for
attackers to find and compromise weak keys.

\section{System Overview}
\label{sec:model}
Our proposed solution to the weak key 
problem, 
pictorially represented in Figure~\ref{fig:proto},
takes place between 
a {\em device}, an {\em entropy authority}, a
{\em certificate authority}\footnote{%
IETF documents~\cite{rfc2527} use the term
{\em certification authority} but we will follow
common usage and use {\em certificate authority}.}
(optionally),
and a {\em client}.
We describe the roles of each of these participants
before outlining our threat model and
the security properties of the scheme.

\subsection{Participants}

\begin{figure}[t]
\centering
\includegraphics[trim=0.06in 0 0 0,clip,width=0.47\textwidth]{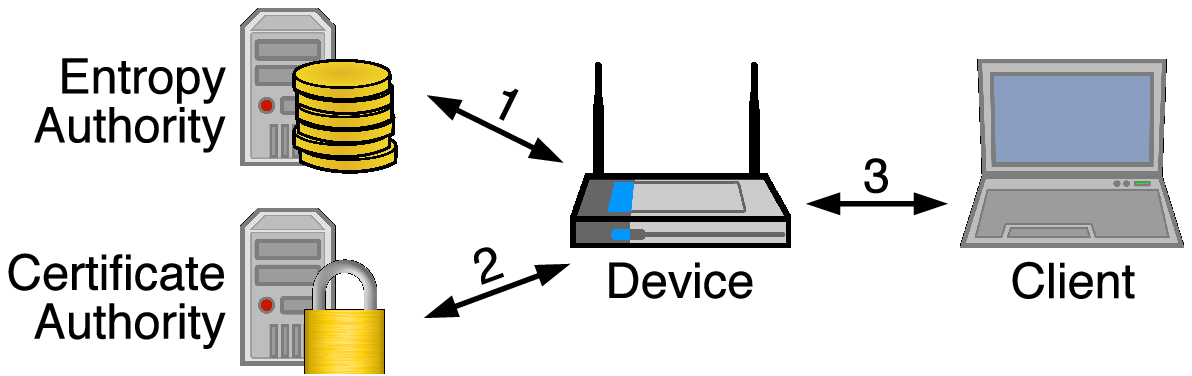} 
\caption{Overview of the protocol participants.
(1) The device fetches random values from the entropy
authority, proves to the authority that its key 
incorporates these values, and obtains a signature
on the key from the EA.
(2) The device submits its EA-signed public key to 
the certificate authority for signing.
(3) The device presents an EA-signed key to
a connecting client to prove that 
its keypair incorporates entropy from the authority.}
\label{fig:proto}
\end{figure}

\paragraph{Device}
The {\em device} is the entity generating
the RSA or DSA keypair that we want to 
ensure is sufficiently random, even if the
device does not have access to a strong
internal entropy source.
The device might be an embedded device
(e.g., a commodity wireless home router), or
it might be a full-fledged server.
The device could use the keypair it generates to
secure HTTPS sessions and to
authenticate itself in SSH sessions. 

\paragraph{Entropy authority (EA)}
The {\em entropy authority} is the participant
responsible for ensuring that a device's keypair
is selected with enough randomness 
(is sampled independently
from a distribution with high enough min-entropy).
Just as a certificate authority verifies the
identifying information (name, address, etc.) 
on a user's public key, the entropy authority
verifies the {\em randomness} of a user's public key.

As the device generates its cryptographic
keypair, it fetches some random values from the
entropy authority and then {\em proves} to
the entropy authority that it has
incorporated these values into its keypair.
The entropy verifies this proof and then 
signs the device's public key.
In practice, an entropy authority is just
a public Web service with which the device interacts 
when the device first generates its keypair.
We assume that the entropy authority has
a strong entropy source, but that the entropy 
authority might be malicious.

We imagine a future in which there are a large number
of public entropy authorities on the Internet, run by 
corporate IT departments, certificate authorities, universities,
and other large organizations.
A device would select its entropy authority much 
as users select certificate authorities today: based on 
reputation and reliability.
To defend against the failure (or maliciousness)
of a single entropy authority, the device
could interact with a number of entropy 
authorities to generate a single key, 
as we describe in Section~\ref{sec:impl}.

\paragraph{Certificate authority (CA)}
The {\em certificate authority} plays the role of
a conventional CA: the certificate authority confirms
that the real-world identity of the device matches
the identity listed in device's certificate, after which
the CA signs the device's certificate.
In our model, CAs will only sign certificates that 
have been signed first by an entropy authority. 
In this way, CAs are guaranteed to sign only public keys
that are drawn from a distribution with high min-entropy.
Since many certificates 
(particularly in embedded devices)
are self-signed, the CA is an optional entity in our protocol.

\paragraph{Client}
The {\em client} is anyone who connects to the device.
In our model, the client can ensure that the
device has a sufficiently random public key by verifying
the entropy authority's signature on the key.
Every client keeps a signature verification
key for each entropy authority it trusts, 
just as today's Web browsers maintain a list of 
public keys for trusted root CAs.

\subsection{Threat Model}
\label{sec:model:threat}

Throughout the paper, 
we say that a participant is
{\em honest} if it performs the protocol
correctly and is {\em dishonest} otherwise.
A device has a {\em strong entropy source} if it can
repeatedly sample from the uniform distribution over
some set (e.g., $\{0, 1\}$).
We say that the device has a {\em weak
entropy source} otherwise.
A {\em strong key}, for our purposes, is
a key independently sampled from a distribution over the
set of possible keys that has
at least $k - \polylog(k)$ bits of 
min-entropy, where $k$ is the security parameter.
In other words, a device generates strong keys if the
probability that the device will generate a particular
public key $\textsf{pk}$ is at most $2^{- (k - \polylog(k))}$ for
all public keys $\textsf{pk}$.
A {\em weak key} is any key that is not strong.

We use min-entropy as our metric for randomness 
because min-entropy gives an upper bound on the 
adversary's ability to guess samples from the distribution.
If a distribution has $k$ bits of min-entropy, then the
adversary will be able to correctly guess the value of 
a sample from the distribution with probability at most $2^{-k}$.
Additionally, the min-entropy of the distribution gives an
upper bound on the {\em collision probability} of the
distribution---the
probability that two independent 
samples from the distribution are equal.
Therefore, if a key generation protocol outputs public keys sampled from a
distribution with min-entropy $k$, two devices generating keys using the
protocol will have a negligible chance (in $k$) of generating the same public
key or of generating RSA keys which share a prime factor.

The goal of our protocol is to have the 
device interact with the entropy authority in 
such a way that, after the interaction, the
device holds a strong cryptographic key.
This overall goal must be tempered a few realities.
In particular, if a device has a no
entropy source (or a very weak entropy source), 
then a {\em global} eavesdropper 
can always learn the device's secret key.

To see why this is so, consider that 
a device with no entropy source is just 
a deterministic process.
Thus, the eavesdropper could always replay such 
a device's interaction with the entropy authority using
the messages collected while eavesdropping.
Thus, there is no hope for a completely 
deterministic device to generate keys
that a global eavesdropper cannot guess.

To circumvent this fundamental problem, we consider
instead a two-phase threat model:
\begin{enumerate}
\item {\em Set-up phase}: 
  In the set-up phase, 
  the device interacts with the entropy authority
  in a communication session 
  that the adversary {\bf cannot observe or modify}.
  In our key-generation protocols, this
  set-up communication session 
  consists of two round-trip interactions
  between the device and the entropy authority.
\item {\em Long-term communication phase}:
  After the set-up stage ends, the adversary 
  can observe and tamper with the traffic on 
  all network links.
\end{enumerate}

This threat model mimics SSH's implicit threat model:
an SSH client gets one ``free'' interaction
with the SSH server, in which the SSH server
sends its public key to the client.
As long as the adversary cannot tamper with 
this initial interaction, SSH protects
against eavesdropping and man-in-the-middle attacks.

Under the adversary model outlined above,
our key generation protocol provides the following
security properties:

\paragraph{Protects device from a malicious EA}
If the device has a strong entropy source,
then the entropy authority learns no useful
information about the device's secrets
during a run of the protocol.
We prove this property 
for the RSA protocol by demonstrating that
the entropy authority can simulate its interaction
with the device given only $O(\log k)$ bits of 
information about the RSA primes $p$ and $q$.
We prove this property 
for the DSA protocol by demonstrating 
that the entropy authority
can perfectly simulate its interaction with the device
given no extra information.

\paragraph{Protects device from CA and client}
An honest device interacting with an honest
entropy authority holds a strong key at 
the end of a protocol run, 
even if the device has a weak entropy source.
When the device later interacts with a certificate
authority (to obtain a public-key certificate) or
with a client (to establish a TLS session), the
device will send these parties a 
{\em strong} public key,
even if the device has weak entropy source.

\paragraph{Protects EA from malicious device}
If the entropy authority is honest, then the
keys generated by this protocol will be strong,
{\em even if} the device is dishonest.
Intuitively, this property states that 
a faulty device cannot discard the random values that
the entropy authority contributes to the key
generation process.

A consequence of this security property is 
that a malicious
device can never ``discredit'' an entropy authority
by tricking the entropy authority into signing a
key sampled from a low-entropy distribution.
If a device does try to have the entropy authority
sign a key sampled from a distribution with
low min-entropy (a {\em weak key}), 
the authority will detect that the
device misbehaved and will refuse to sign the key.

A nuance of this property is that the entropy authority
{\em will accept public keys that are invalid}, 
as long as the keys are sampled independently from
a distribution with high min-entropy.
In essence, a faulty device in our protocol
can create keys that are incorrect but random.
For example, the device could pick an composite number
as one of its RSA ``primes,'' 
or it could use any number of
other methods to ``shoot itself in the foot''
during the key generation process.
Since the device can {\em always} compromise its
own keypair (e.g., by publishing its secret key), we
do not attempt to protect a completely malicious device from itself.
Instead, we simply guarantee that any key that the 
entropy authority accepts will be drawn
independently from a distribution with high min-entropy. 

\subsection{Non-threats}
Our protocol addresses the threat posed by
devices that use weak entropy sources to generate
their cryptographic keys.
We explicitly {\em do not} address these other broad
vulnerability classes:
\begin{itemize}
\item \textbf{Adversarial devices.}
      If the device is completely adversarial, then 
      the device can easily compromise its own security (e.g.,
      by publishing its own secret key).
      Ensuring that such an adversarial device has
      high-entropy cryptographic keys is not useful,
      since {\em no} connection to such an adversarial
      device is secure.

\item \textbf{Faulty cryptography library (or OS).}
      Our protocol does not attempt to protect
      against cryptographic software that 
      is arbitrarily incorrect. 
      Incorrect software can introduce any number of
      odd vulnerabilities 
      (e.g., a timing channel that leaks the secret key),
      which we place out of scope.

\item \textbf{Denial of service.}
      We do not address denial-of-service attacks
        by the entropy authority or certificate authority. 
        In a real-world deployment, we expect that
        a device facing a denial-of-service attack 
        by a CA or entropy authority 
        could simply switch to using a new CA or EA.

\end{itemize}

\section{Protocol}
\label{sec:proto}
This section describes a number of standard cryptographic
primitives we require and then outlines our RSA and DSA
key generation protocols.

\subsection{Preliminaries}
Our key generation protocols use the following
cryptographic primitives.

\paragraph{Additively homomorphic commitments}
We require an additively homomorphic and
perfectly hiding commitment scheme.
Given a commitment to $x$ and a commitment to
$x'$, anyone should be able to construct
a commitment to $x+x' \pmod Q$ without knowing
the values $x$ or $x'$.
Our implementation uses Pedersen
commitments~\cite{pedersen92noninteractive}.
Given public generators $g,h$ of a group
$G$ with prime order $Q$,
and a random value $r \in \mathbb{Z}_Q$, a Pedersen
commitment to the value $x$ is
$\textsf{Commit}(x; r) = g^x h^r$.%
\footnote{
We denote the group order with capital ``$Q$'' 
to distinguish it from the
RSA prime $q$ in $n=pq$ that we use later on.} 
To ensure that the commitments are binding, 
participants must select the generators $g$ and $h$ 
in such a way that {\em no one} knows the discrete
logarithm $\log_g h$.

The commitment scheme is additively homomorphic
because the product of two commitments reveals
a commitment to $x+x' \pmod Q$ with randomness 
$r+r' \pmod Q$:
\begin{align*}
\textsf{Commit}(x+x'; r+r') = \textsf{Commit}(x; r)\textsf{Commit}(x'; r)
\end{align*}
We abbreviate $\textsf{Commit}(x; r)$ as 
$\textsf{Commit}(x)$ when the randomness used
in the commitment is not relevant to the exposition.

Of course, if the device has a weak entropy source the device will not be able
to generate a strong random value $r$ for use in the commitments.
We use randomized commitments to hide a device's secrets {\em in case}
the device does have a strong entropy source.
Since a device does not necessarily know whether its randomness source
is strong or weak, we must use the same constructions for devices
with both strong and weak entropy sources.

\paragraph{Public-key signature scheme}
We use a standard public-key signature scheme
that is existentially unforgeable~\cite{goldwasser88digital}.
We denote the signing and verification algorithms
by $\textsf{Sign}$ and $\textsf{Verify}$.

\paragraph{Multiplication proof for committed values}
We use a zero-knowledge proof-of-knowledge protocol that
proves that the product of two committed values is
equal to some third value.
For example, given commitments $C_x$ and $C_y$ to
values $x,y \in \mathbb{Z}_Q$, and a 
third product value $z \in \mathbb{Z}_Q$,
the proof demonstrates that $z = xy \pmod Q$.
We denote the prover and verifier algorithms
by $\pi \leftarrow \textsf{MulProve}(z, C_x, C_y)$ and
$\textsf{MulVer}(\pi, z, C_x, C_y)$.

We implement this proof using the method of
Cramer and Damg{\aa}rd~\cite{cramer98zero}. 
Written in Camenisch and Stadler's zero-knowledge proof
notation~\cite{camenisch97proof}, 
the multiplication proof proves the statement:
\begin{align*}
  \textsf{PoK}\{x, y, r_x, r_y, r_z:& \\
    C_x = g^x h^{r_x} &\land C_y = g^y h^{r_y} \land g^z h^{r_z} = (C_x)^y h^{r_z}\}
\end{align*}
Application of the Fiat-Shamir heuristic~\cite{fiat86prove}
converts this interactive zero-knowledge proof protocol
into a non-interactive proof in the random-oracle
model~\cite{bellare93random}.
When implemented using a hash function that outputs
length-$l$ binary strings, the non-interactive 
multiplication proof is $l + 3 \lceil \log_2 Q \rceil$ bits long.

\paragraph{Proof of knowledge for Pedersen commitments}
We use a non-interactive zero-knowledge proof-of-knowledge protocol that proves
that a committed value is equal to the discrete logarithm of a second group
element. 
Given a Pedersen commitment $C_x = g^x h^r$, a DSA public key $A$, 
and an auxiliary value $x'$, the proof demonstrates that:
\[ \textsf{PoK}\{ x, r : C_x = g^x h^r \land (A/g^{x'}) = g^x \} \]
We denote the prover and verifier algorithms
by $\pi \allowbreak \leftarrow \allowbreak \textsf{PedProve}(\allowbreak{}x, r, C_x, x', A)$ and
$\textsf{PedVer}(\pi, C_x, x', A)$.

We implement this proof using the method of Camenisch and
Stadler~\cite{camenisch97proof} and we apply the Fiat-Shamir
heuristic~\cite{fiat86prove} to convert the interactive proof into a
non-interactive proof in the random-oracle model~\cite{bellare93random}.
When implemented using a hash function that outputs
length-$l$ binary strings, the non-interactive 
proof is $l + 2 \lceil \log_2 Q \rceil$ bits long.

\paragraph{Common Public Keys}
We assume that all participants hold a signature 
verification public-key
for the entropy and certificate authorities.

\subsection{RSA Key Generation}
\label{sec:RSAgen}

The RSA key generation protocol takes place
between the device and the entropy authority.
At the end of a successful run of the
protocol, the device holds
an RSA public modulus $n$ that
is independently sampled from a distribution over
$\mathbb{Z}$ that has high min-entropy
and the device also holds the entropy
authority's signature $\sigma$ on this modulus.

In Section~\ref{sec:sec:rsa} we prove that the
RSA protocol satisfies the security properties defined in
Section~\ref{sec:model:threat}.
In Section~\ref{sec:impl}, we describe how a
device could use this protocol to generate
a self-signed X.509 certificate and how to 
integrate this protocol with today's certificate
authority infrastructure. 

\paragraph{Parameters}
Before the protocol begins, the device and
entropy authority must agree on a set of
common system parameters. 
These parameters include the security 
parameter $k$, which determines the
bit-length of the RSA primes $p$ and $q$.
For a given value of $k$, the
participants must also agree on a 
prime-order group $G$ used for 
the Pedersen commitments and zero-knowledge proofs.
The prime order $Q$ of the group $G$ must be somewhat larger
than the largest RSA modulus~$n$ generated
by the protocol, so the participants should 
let $Q \approx 2^{2k+100}$.
In addition, participants must agree on two
generators $g$ and $h$ of the group $G$, such that 
{\em no one} knows the discrete logarithm $\log_g h$.
In an implementation of the protocol, 
participants could generate $g$ and $h$ 
using a shared public hash function.
Finally, they also agree on a small number $\Delta$ (e.g., $\Delta = 2^{16}$)
discussed in Section~\ref{sec:findingprimes} below.

Since the parameters contain only public values,
all devices and entropy authorities could share
one set of parameters (per key size).

%%%%%%%%%%%%%%%%%%%%%%%%%%
% These commands are used in the
% protocol figures below

\newcommand{\figureArrow}[4]{
\multicolumn{3}{c}{
  \begin{picture}(80,20)
  \put(40,8){\makebox(0,0){#1}}
  \put(#2,0){\vector(#3,#4){80}}
  \end{picture}
}\vspace{-10pt}}

\newcommand{\figureRightArrow}[1]{\figureArrow{#1}{0}{1}{0}}
\newcommand{\figureLeftArrow}[1]{\figureArrow{#1}{80}{-1}{0}}

\newcounter{protoCounter}

\newcommand{\ProtoStep}{\refstepcounter{protoCounter}%
  \underline{Step \arabic{protoCounter}}}

\begin{figure}
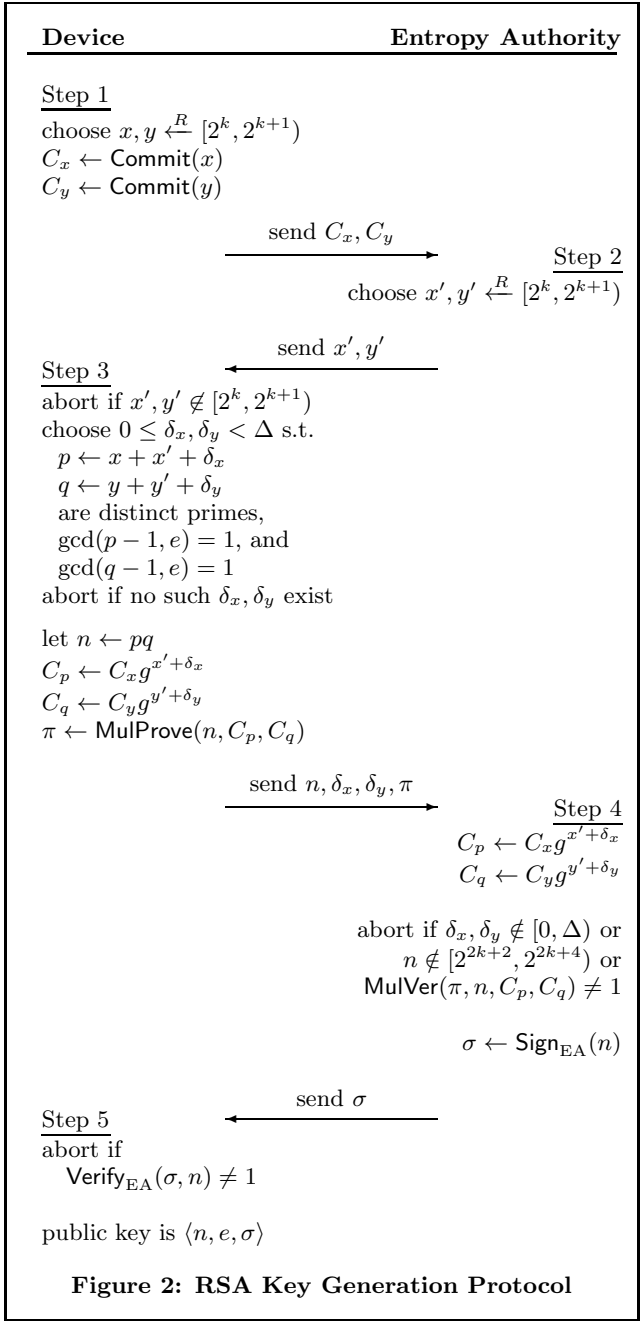

\begin{framed}
\centering
\begin{tabular*}{\textwidth}{l @{\extracolsep{\fill}} c r}
\textbf{Device} & & \textbf{Entropy Authority}\\
\hline
\\
\ProtoStep\label{rsa:commit}\\
choose $x,y \xleftarrow{R} [2^k, 2^{k+1})$  \\
$C_x \leftarrow \textsf{Commit}(x)$  \\
$C_y \leftarrow \textsf{Commit}(y)$  \\

\figureRightArrow{send $C_x, C_y$}\\[1mm]

&&\ProtoStep\label{rsa:ea}\\
& \multicolumn{2}{r}{choose $x',y' \xleftarrow{R} [2^k, 2^{k+1})$}\\[1mm]

\figureLeftArrow{send $x', y'$}\\[1mm]

\ProtoStep\label{rsa:delta}\\
abort if $x',y' \not\in [2^k, 2^{k+1})$ \\
choose $0 \leq \delta_x, \delta_y < \Delta $ s.t.\\
\ \ $p \leftarrow x+x'+\delta_x$\\
\ \ $q \leftarrow y+y'+\delta_y$\\
\ \ are distinct primes, \\
\ \ $\gcd(p-1, e) = 1$, and\\
\ \ $\gcd(q-1, e) = 1$\\
abort if no such $\delta_x,\delta_y$ exist\\[2mm]
let $n \leftarrow  pq$\\ 
$C_p \leftarrow  C_x g^{x' + \delta_x}$\\
$C_q \leftarrow  C_y g^{y' + \delta_y}$\\
$\pi \leftarrow \textsf{MulProve}(n, C_p, C_q)$\\[1mm]

\figureRightArrow{send $n, \delta_x, \delta_y, \pi$}\\[1mm]

&&\ProtoStep\label{rsa:easign}\\
&&$C_p \leftarrow  C_x g^{x'+\delta_x}$\\
&&$C_q \leftarrow  C_y g^{y'+\delta_y}$\\ 
\\
%&&abort if $\delta_x,\delta_y \notin [0, \Delta)$ or\\
&\multicolumn{2}{r}{abort if $\delta_x,\delta_y \notin [0, \Delta)$ or}\\
&&$n \notin[2^{2k+2}, 2^{2k+4})$ or\\
&&$\textsf{MulVer}(\pi, n, C_p, C_q) \neq 1$\\
\\
&&$\sigma \leftarrow \textsf{Sign}_\textrm{EA}(n)$\\[1mm]

\figureLeftArrow{send $\sigma$}\\[1mm]

\ProtoStep\\
abort if \\
\quad$\textsf{Verify}_\textrm{EA}(\sigma, n) \neq 1$\\
\\
public key is $\langle n, e, \sigma \rangle$\\
\end{tabular*}
\caption{RSA Key Generation Protocol}
\label{fig:proto-rsa}
\end{framed}
\end{figure}

\paragraph{Protocol Description}
Figure~\ref{fig:proto-rsa} presents our RSA key 
generation protocol.
To generate an RSA key, the device first selects
$k$-bit integers $x$ and $y$ and
sends randomized commitments to these values
to the entropy authority.
The entropy authority then selects $k$-bit
integers $x'$ and $y'$ at random and returns these values to the
device.

After confirming that $x'$ and $y'$ are of the 
correct length, the device searches for offsets
$\delta_x$ and $\delta_y$ such that the sums
$p=x+x'+\delta_x$ and $q=y+y'+\delta_y$ are suitable
RSA primes.
That is, $p$ and $q$ must be distinct primes such
that $\gcd(p-1, e) = 1$ and $\gcd(q-1,e) = 1$,
where $e$ is the RSA encryption exponent.
The device then sets $n \gets pq$, 
generates commitments to $p$ and $q$,
and produces a non-interactive zero-knowledge proof of
knowledge $\pi$ that the product of the committed values
is equal to $n$.
The device sends $n$, $\delta_x$, $\delta_y$, and the
the proof $\pi$ to the entropy authority.

The validity of the proof $\pi$ and the fact
that the $\delta$ values are less than
$\Delta$ convince the entropy authority 
that the device's RSA primes $p$ and $q$ incorporate
the authority's random values $x'$ and $y'$. 
At this point, the authority signs the modulus $n$ and
returns it to the device.

\subsubsection{Finding Primes $p$ and $q$}
\label{sec:findingprimes}

To maintain the security of the protocol, 
it is important that the $\delta$ values 
chosen in Step~\ref{rsa:delta} are 
relatively small---if the device could pick
an arbitrarily large $\delta_x$ value, for example,
the device could set $\delta_x \leftarrow -x'$,
which would make $p = x + x' - x' = x$, 
thereby cancelling out the effect of the 
random value $x'$ contributed by 
the entropy authority.
To prevent the device from ``throwing away''
the entropy authority's entropy in this way,
we require that the $\delta$ values be less than
some maximum value $\Delta$, which depends on the security parameter $k$.

Picking the size of $\Delta$ requires some care:
if $\Delta$ is too small, then there may be no suitable
prime $p$ in the range $[x+x', x+x'+\Delta)$,
and the device will have to run the protocol 
many times before it finds suitable primes $p$ and $q$.
The value $\Delta$ should be large enough that 
the protocol will succeed with overwhelming 
probability, but not so large that the device
can pick $n = pq$ arbitrarily.

Following Juels and Guajardo~\cite{juels02verifiable}, 
if the density of primes is $d_\textrm{prime}$
and the density of these special primes (with $\gcd(p-1,q-1,e) = 1$)
is $d_\textrm{special}$, 
we conjecture that $d_\textrm{special}/d_\textrm{prime} = (e-1)/e$,
where $e$ is the RSA encryption exponent (a small odd prime constant).
Under this conjecture and the Hardy-Littlewood~\cite{hardy23some}
conjecture, Juels and Guajardo demonstrate that the probability
that there is {\em no suitable prime} in the interval $[x+x', x+x'+\Delta)$
is at most $\exp(-\lambda)$ when $\Delta = \lambda \ln(x+x') (\frac{e}{e-1})$
as $(x+x') \rightarrow \infty$.
To make this conjecture concrete: if we take $(x+x') \approx 2^{1024}$, the
RSA encryption exponent $e = 65537$, and require a failure
probability of at most $2^{-80}$, then we should set 
$\Delta \approx 2^{16}$.
In the very unlikely case that the device fails to find
primes $p$ and $q$ in the right range, the device aborts
and re-runs the protocol from the beginning.

\subsubsection{Eliminating Information Leakage}
\label{sec:elim}

The values $\delta_x$ and $\delta_y$ sent to the 
entropy authority in Step~\ref{rsa:delta} of the protocol
leak some information about $p$ and $q$ to the 
entropy authority.
In particular, the authority learns that 
the prime gap before $p$ (resp. $q$) has a width of
at least $\delta_x$ (resp. $\delta_y$).
We argue in Section~\ref{sec:sec:rsa} the 
entropy authority cannot use this leakage to help it
factor the modulus $n$.

Even so, it is possible to modify the protocol
to completely eliminate this information leakage
at some performance cost.
One way to modify the protocol is to require
that $\delta_x = \delta_y = 0$ in Step~\ref{rsa:delta}
of the protocol.
If the values $x+x'$ and $y+y'$ are not prime, the device
aborts the protocol and restarts it from the beginning.
Since the probability that a random $k$-bit number is
a suitable prime is near $1/k$ for large $k$, the device will
have to run the protocol approximately $k^2$ 
times before it succeeds.

To reduce the number of communication
rounds required for this revised protocol, 
the device could run the $k^2$ protocol iterations in parallel.
The device would send $k^2$-length vectors
of commitments to random values $\vec{x}, \vec{y}$ in 
Step~\ref{rsa:commit} of the protocol and the
entropy authority would return two vectors $\vec{x}', \vec{y}'$
in Step~\ref{rsa:ea} of the protocol, with each
vector having length $k^2$.
The device would then iterate over the vectors
until it finds an $i$ such that
$p \gets x_i+x'_i$ and $q \gets y_i+y'_i$ 
are distinct primes and $\gcd(p - 1, q - 1, e) = 1$.
If the device fails to find such primes, it would
abort and repeat the process.

\subsection{DSA Key Generation}
The DSA key generation protocol, 
which we present in Figure~\ref{fig:proto-dsa}, takes 
place between a device and the entropy authority.

\paragraph{Parameters}
We assume that, before the start of the
protocol, participants have agreed
upon an order-$Q$ group $G$ used in the
DSA signing process.
If the device uses the elliptic-curve variant of DSA (EC-DSA),  
then the group $G$ will be an
elliptic curve group selected, for example,
from one of the NIST standard curves~\cite{gallagher09fips}.
Participants must also agree upon two public generators,
$g$ and $h$, of the group $G$ such that {\em no one}
knows the discrete logarithm $\log_g h$.

While we expect most new devices to primarily use EC-DSA keys,
even new devices may also need to generate finite-field DSA
keys for interoperability with legacy devices.
When using the finite-field variant of DSA,
the device may have to generate the parameters
of the finite-field DSA group (a prime modulus $p$,
a group order $Q$, and a generator $g$) 
in addition to its keypair.
To do this, the device and entropy authority could
agree on a {\em domain parameter seed} using
a coin-flipping protocol~\cite{blum83coin} 
and then use this shared seed to 
generate DSA parameters using the verifiable 
generation method specified in the Digital Signature
standard~\cite[Appendix A]{gallagher09fips}.

\paragraph{Protocol Description}
To begin the key generation process
depicted in Figure~\ref{fig:proto-dsa}, the device
picks a random value $x \in \mathbb{Z}_Q$ and
generates a randomized commitment to $x$.
In the event that the device has a strong entropy source,
the use of a randomized commitment prevents
the entropy authority from learning the device's
secret $x$.
The device sends this commitment to the entropy
authority.

Upon receiving the device's commitment, the entropy
authority chooses a random value $x' \in \mathbb{Z}_Q$ 
and returns this value to the device.
The device sets its private key $a \gets x+x' \mod Q$
and sets its public key to $A \gets g^a$.
The device then sends its public key $A$ along with a 
non-interactive proof of correctness $\pi$ to the entropy
authority.

The entropy authority verifies the proof $\pi$, which convinces
the entropy authority that $A$ is equal to $g^{x+x'}$.
The entropy authority then signs the device's public
key $A$ and returns it to the device.

\setcounter{protoCounter}{0}
\begin{figure}
\begin{framed}
\centering
\begin{tabular*}{\textwidth}{l @{\extracolsep{\fill}} c r}
\textbf{Device} & & \textbf{Entropy Authority}\\
\hline \\

\ProtoStep\\
choose $x, r \xleftarrow{R} \mathbb{Z}_Q $ & & \\
$C_x \leftarrow \textsf{Commit}(x; r)$ & & \\

\figureRightArrow{send $C_x$}\\[1mm]

&&\ProtoStep\\
&&choose $x' \xleftarrow{R} \mathbb{Z}_Q $\\[1mm]

\figureLeftArrow{send $x'$}\\[1mm]

\ProtoStep\label{dsa:pub}\\
$a \leftarrow x + x' \bmod Q$\\
$A \leftarrow g^a$\\
\multicolumn{3}{l}{ 
$\pi \leftarrow \textsf{PedProve}(x, r, C_x, x', A)$}\\

\figureRightArrow{send $A, \pi$}\\[1mm]

&&\ProtoStep\label{dsa:eacheck}\\
  && abort if\\
  \multicolumn{3}{r}{$\textsf{PedVer}(\pi, C_x, x', A) \neq 1$}\\
  \\
&&$\sigma \leftarrow \textsf{Sign}_\textrm{EA}(A)$\\

\figureLeftArrow{send $\sigma$}\\[1mm]

\ProtoStep\\
abort if $\textsf{Verify}_\textrm{EA}(\sigma, A) \neq 1$&&\\
\\
public key is $\langle A, \sigma \rangle$\\

\end{tabular*}
\caption{DSA Key Generation Protocol}
\label{fig:proto-dsa}
\end{framed}
\end{figure}

\newcommand{\mypar}[1]{\medskip \noindent {\bf #1.}}
\newcommand{\adv}{{\mathcal A}}
%\newcommand{\deq}{\mathrel{\mathop:}=}
% Using \gets for consistency with rest of paper
\newcommand{\deq}{\mathrel\gets}

\section{Security Analysis}
\label{sec:sec}

This section presents proofs that 
the RSA and DSA key generation
protocols satisfy the security properties
described in Section~\ref{sec:model}.

\subsection{RSA Protocol}
\label{sec:sec:rsa}

\subsubsection{Protects Device from a Malicious EA}

We first show that when the device has a strong entropy
source, a malicious EA learns no useful information
about the device's resulting RSA secret key.

First, let us define a standalone RSA modulus generation
algorithm which does not interact with an EA.   
The key generator takes as input a security parameter $k$ and 
lower bounds $p_{\text{min}}$ and $q_{\text{min}}$ 
on the RSA primes $p$ and $q$.

\begin{minipage}{\columnwidth}
\begin{tabbing}
0000 \= 0000 \= 0000 \= \kill
\>$\textsf{PrimeGen}(k,p_{\text{min}})$: \\
\>\>  choose a random $x$ in $[2^k, 2^{k+1}]$ \\
\>\>  find the smallest prime $p$ s.t. $p \geq p_{\text{min}}+x$  \\
\>\>\> and s.t. $\gcd(p-1, e) = 1$\\
\>\>  output $p$ \\[2mm]
\>$\textsf{RSAKeyGen}(k,p_{\text{min}},q_{\text{min}})$:  \\
\>\>  $p \gets \textsf{PrimeGen}(k,p_\textrm{min})\ , \ q \gets \textsf{PrimeGen}(k, q_\textrm{min})$ \\
\>\>  output $n \gets p \cdot q$
\end{tabbing}
\end{minipage}
\\
%\vspace{\baselineskip}

We say that a modulus generator outputs a {\em secure
distribution of RSA moduli~$n$} if the resulting family of RSA functions
$x \to x^e \bmod n$ 
is a family of trapdoor one-way functions
(where $e$ is the RSA encryption exponent, 
a small prime constant).
A secure
modulus generator is sufficient for use in standard RSA encryption and RSA
signature systems.

We use the following {\em RSA assumption} about algorithm
$\textsf{RSAKeyGen}$ above: algorithm
$\textsf{RSAKeyGen}(k,p_{\text{min}},q_{\text{min}})$ outputs a
secure distribution of RSA moduli for all $p_{\text{min}}$ and
$q_{\text{min}}$ in the interval $[2^k, 2^{k+1})$.

The following theorem shows that even when interacting with a
malicious EA, the RSA key generation protocol in
Figure~\ref{fig:proto-rsa} outputs a secure distribution of RSA
moduli.  Furthermore, the protocol leaks at most $O(\log k)$ bits of
information about the prime factors to the EA.  This small leak does
not harm security since if it were possible to invert the RSA function
given the few leaked bits then it would also be possible to do it
without, simply by trying all possible values for the leaked bits
in time polynomial in $2^{\log k} = k$.
Moreover, if desired this small leak can be eliminated at the cost of
more computation, as explained in Section~\ref{sec:elim}.

\begin{theorem}
\label{thm:rsa-device}
Suppose the device has a strong entropy source (i.e., the device can
repeatedly sample independent uniform bits in $\{0,1\}$).  Then for all EA, the
protocol in Figure~\ref{fig:proto-rsa} generates a secure
distribution of RSA moduli assuming the RSA assumption above.
Furthermore, EA's view of the protocol can be simulated with
at most $O(\log k)$ advice bits with high probability.
\end{theorem}

\begin{proof}
Let ${\adv}$ be a malicious EA that, given random
commitments $C_x,C_y$, outputs $(x',y') \gets \adv(C_x,C_y)$.
Then, since Pedersen commitments are information theoretically
hiding, the protocol in Figure~\ref{fig:proto-rsa} outputs a 
modulus $n$ sampled from the following distribution:
\begin{tabbing}
0000 \= 0000 \= \kill
\> choose random $C_x, C_y \xleftarrow{R} \mathbb{Z}_Q$ \\
\> $(p_{\text{min}},q_{\text{min}}) \gets \adv(C_x,C_y)$ \\
\> output $\textsf{RSAKeyGen}(k,p_{\text{min}},q_{\text{min}})$
\end{tabbing}
Therefore, by the RSA assumption about 
algorithm $\textsf{RSAKeyGen}$ the protocol 
generates a secure distribution of RSA moduli.

Next, to argue that the protocol leaks at most $O(\log k)$ bits of
information about the prime factors with high probability,
we construct a simulator~$S$
that simulates the transcript of a successful run of the protocol with
$\adv = \langle \adv_1, \adv_2 \rangle$ 
given only~$n$ and an additional $O(\log k)$ bits of
information.  This will prove that given $n$, the protocol leaks only
$O(\log k)$ additional bits.
The protocol transcript consists of 
\[  \langle C_x,C_y, x', y', n, \delta_x, \delta_y, \pi, \sigma \rangle \]
where $n = pq$ and $p = x+x'+\delta_x,\ \ q=y+y'+\delta_y$ for some $x,y$.
For a prime $p$ let $\textsf{pre}(p)$ be the prime immediately preceding~$p$,
such that $\textsf{pre}(p)-1$ is relatively prime to the RSA encryption exponent $e$. 
The simulator $S$ takes three arguments as input:
the modulus $n=pq$ produced by a successful run of the protocol
and the quantities 
\[  \Delta_p = \text{min}(p - \textsf{pre}(p),\ \Delta)
       \quad\text{;}\quad
    \Delta_q = \text{min}(q - \textsf{pre}(q),\ \Delta)
\]

The simulator works as follows:
\begin{tabbing}
000 \= 000 \= \kill
$S(n,\Delta_p, \Delta_q)$: \\
\> repeat: \\
\>\> set $C_x, C_y \xleftarrow{R} \mathbb{Z}_Q$ \\
\>\> set $(x',y') \gets \adv_1(C_x, C_y)$ \\
\> until $n \in \left[(x'+2^k)(y'+2^k),\ (x'+2^{k+1})(y'+2^{k+1})\ \right)$ \\
\> set $\delta_x \xleftarrow{R} [0,\Delta_p), \; \delta_y \xleftarrow{R} [0,\Delta_q)$ \\
\> use the NIZK simulator to simulate a proof $\pi$\\
\>\>    that $n = (x+x'+\delta_x)(y+y'+\delta_y)$ \\
\>\>    where $x$ and $y$ are the values committed in $C_x,C_y$\\
\> set $\sigma \gets \adv_2(n, \delta_x, \delta_y, \pi)$\\
\> output the simulated transcript: \\
\>\>    $\langle C_x,C_y, x', y', n, \delta_x, \delta_y, \pi, \sigma \rangle$
\end{tabbing}
The simulator $S$ properly simulates the Pedersen commitments $C_x,C_y$
and the quantities $x',y'$, given that the protocol generated the modulus $n$.
Similarly, given that $n=pq$ was the output we know that the random
variable $x+x'$ is uniformly distributed in the interval $(\textsf{pre}(p),p]$
whenever $p-\text{pre}(p)<\Delta$ and is uniform in $(p-\Delta,p]$ otherwise.
Either way, the value of $\delta_x$ is uniform in $[0,\Delta_p)$.
Hence $S$ properly simulates $\delta_x$ and similarly $\delta_y$.   
Finally, $\pi$ is properly simulated using the ZK knowledge simulator for 
a proof of Pedersen products. 

%%% Need to explain better:  
%%%   - why does first while loop in S complete in expected poly. time.
%%%   - is the condition for terminating the loop sufficient for
%%%     ensuring the existence of x',y'.

We explained in Section~\ref{sec:RSAgen} that $\Delta_p$ and $\Delta_q$
are $O(k)$ in size, and therefore the protocol leaks 
at most $O(\log k)$ bits of information 

\medskip

An important technical point is that 
a malicious entropy authority could send
an invalid signature $\sigma$ in the last step of the protocol,
which would cause the device to abort the protocol.
If the entropy authority allows the protocol to complete with only negligible
probability, then the simulator will have to rewind the entropy authority 
a super-polynomial number of times, and the simulator will not necessarily
succeed in polynomial time.

If an entropy authority only allows the protocol to complete with negligible
probability, however, a device will only be able to generate a key using one such
entropy authority with negligible probability. 
Thus a device is extremely unlikely to ever use such a key in practice.
Against adversaries that {\em do} allow the protocol to complete with
non-negligible probability, the simulator will always run in expected polynomial time.
\end{proof}

\subsubsection{Protects Device from the CA and Client}

Having established that the protocol protects a high-entropy
device from the entropy authority, we demonstrate that
an honest device interacting with an honest
entropy authority holds a strong key at 
the end of a protocol run, even if the device has a weak entropy source.

\begin{theorem}
When interacting with an honest EA, the RSA protocol in
Figure~\ref{fig:proto-rsa} generates a secure distribution of RSA
moduli, assuming that algorithm $\textsf{RSAKeyGen}$ outputs
a secure distribution of RSA moduli.
\end{theorem}
\begin{proof}
Let $\adv$ be a device honestly following the protocol, but one that
may have a weak entropy source.  We let $\adv()$ denote the $x,y$ chosen by the
device in Step~\ref{rsa:commit}.  Given an honest EA, the protocol in
Figure~\ref{fig:proto-rsa} outputs a modulus~$n$ sampled from the
following distribution:
\begin{tabbing}
0000 \= 0000 \= \kill
\> $(p_{\text{min}},q_{\text{min}}) \gets \adv()$ \\
\> output $\textsf{RSAKeyGen}(k,p_{\text{min}},q_{\text{min}})$
\end{tabbing}
By the RSA assumption about algorithm $\textsf{RSAKeyGen}$ the protocol 
generates secure RSA moduli.
\end{proof}

\subsubsection{Protects EA from a Malicious Device}
\label{sec:security}

\abbr{
%%%%%%%   camera ready text   %%%%%

Suppose the device is {\em dishonest} and its goal is to discredit the
entropy authority.  The device may try to cause the EA to sign a
modulus $n$ in Step~\ref{rsa:easign} of the protocol where $n$ is sampled from a low
entropy distribution.  For example, the two prime factors of $n=pq$
may look non-random (e.g. their binary representation may end in many
1's) or $n$ may have a non-trivial GCD with another public RSA
modulus.  The EA's signature would then serve as incriminating evidence
that the ``random'' values $x'$ and $y'$ the EA contributed
to the protocol in Step~\ref{rsa:ea} were not sampled 
from the uniform distribution over $[2^k, 2^{k+1})$.

Note, however, that if the modulus $n$ output by the device is an
ill-formed RSA modulus---say $n$ is not a product of two primes---%
then clearly the EA is not at fault since the device did not properly
generate $n$.  Therefore the EA need not worry about invalid moduli.
It only cares about not signing low-entropy moduli.

The following theorem shows that an honest EA will never sign 
a low-entropy modulus.   

\begin{theorem}
Consider an honest entropy authority interacting with a
malicious polynomial-time device.  Suppose that the RSA modulus $n$ signed 
by an honest entropy authority in Step~4 is a product of two distinct
primes $n=pq$ each in the range $[2^{k+1},2^{k+2})$.  
Then each of the primes is sampled
from a distribution with at least $k-2 - d \log(k)$ bits of
min-entropy for some absolute constant $d$, even when
conditioned on the other prime.
\end{theorem}

{\sc Proof sketch.}
We first show that $n$ must be sampled from a distribution with sufficiently
high min-entropy.
In Step~4 of the protocol, the proof $\pi$ convinces the EA that
\[   n = (x+x'+\delta_x)(y+y'+\delta_y) \pmod{Q}  \]
for some (unknown) $x$ and $y$,
where $Q$ is the group order used for the Pedersen commitments.  Recall
that $Q>2^{100} 2^{2k}$.  Since the device must commit to $x$ before
seeing $x'$ we know that $x+x'$ is sampled independently from a
distribution over $\mathbb{Z}_Q$ with min-entropy at least $k$ (in the
worst case, $x+x'$ is sampled uniformly from the integers in an interval 
of width $2^k$).  Since
the device controls $\delta_x$ and $0 \leq \delta_x < \Delta < c\cdot k$ for
some absolute constant $c$, it follows that the min-entropy of $p_0
\deq x+x'+\delta_x$ is at least $k-\log c k$.  Similarly the
min-entropy of $q_0 \deq y+y'+\delta_y$ conditioned on $p_0$ is at
least $k-\log c k$.  Consequently, since $n$ is a product
of two primes, it can be shown for the
distributions in question here that the min-entropy of $n = p_0 q_0
\bmod Q$ is at least $2k-d \log k$ for some constant $d$.

%%% The statement above will follow from a number theoretic fact
%%% saying that for all x,y such that (x+x'+delta_x)(y+y'+delta_y) falls in 
%%% the right range with non-negligible probability, and for all RSA moduli n,
%%% the # of small points (x',y') such that  (x+x')(y+y') = n \bmod Q
%%% is small.

If $n$ is the product of two primes $n=pq$ each in
the range $[2^{k+1},2^{k+2})$ then each prime must be chosen 
from a distribution with min-entropy at least $k-2-d \log k$
\ (otherwise $n$ cannot have min-entropy at least $2k-d \log k$).
The theorem now follows.
\qed

%%%%%%%   end camera ready text   %%%%%
}{
Suppose the device is {\em dishonest} and its goal is to discredit the
entropy authority.  
The device may try to cause the EA to sign a
modulus $n$ in Step~\ref{rsa:easign} of the protocol, such that $n$ is
sampled from a low-entropy distribution.  For example, one of the
prime factors of $n=pq$ may look non-random (e.g. its binary
representation may end in many 1's), or $n$ may have a non-trivial GCD
with another public RSA modulus, or the two prime factors~$p$ and~$q$
may not be sampled independently, say $q = p+2$.  If the factors $p$
and $q$ become public then the EA's signature would serve as
incriminating evidence that the ``random'' values $x'$ and $y'$ the EA
contributed to the protocol in Step~\ref{rsa:ea} were not sampled
independently from the uniform distribution over $[2^k, 2^{k+1})$.

If the modulus $n$ output by the device is an ill-formed RSA modulus,
say $n$ is not a product of two primes or the primes are not in
$[2^{k+1}, 2^{k+2})$ then clearly the EA is not at fault since the device
did not properly generate $n$.  Therefore the EA need not worry about
invalid moduli---it should only care about not signing moduli
sampled from low-entropy distributions. 

We argue in Theorem~\ref{thm:malice} below that an honest EA will never sign a
modulus sampled from a low-entropy distribution even when interacting with a
malicious device.  
The desire to protect the EA from a malicious
device explains why we need a protocol such as the protocol of
Figure~\ref{fig:proto-rsa}.  If this property is not needed then a far
simpler protocol is sufficient: the EA can simply send a random string
to the device.  The device will generate a modulus $n$ incorporating
the entropy from the EA and send the resulting modulus back to the EA
for signing.  The problem is that in this trivial protocol the EA
blindly signs the modulus without any guarantees that the modulus incorporates
the EA's randomness.  Indeed, a dishonest device could easily get
the EA to sign a modulus $n=pq$ where the primes $p$ and $q$ are
selected from a low-entropy distribution, say where $p=q+2$.  The
dishonest device could then claim that the EA provided faulty entropy
and thereby discredit the EA.

\medskip
Theorem~\ref{thm:malice} below shows that no malicious device can
discredit the EA when the protocol in Figure~\ref{fig:proto-rsa} is
used.  In particular, we show that a modulus $n$ signed by the EA in
Step~\ref{rsa:easign} is sampled from a distribution with min-entropy
of at least $2k - \polylog(k)$ bits.  
Since we assume $n$ is well
formed, the primes $p$ and $q$ must lie in the interval $[2^{k+1}, 2^{k+2})$ 
and therefore each prime must be chosen from a distribution
with independent min-entropy of about $k$ bits (otherwise $n$ cannot
have min-entropy of around $2k$ bits).  
% The fact that the primes have
% min-entropy of about $k$ bits has two positive implications:
% \begin{itemize}
% \item First, if the min-entropy is about $k$ bits then the probability
% that an adversary can successfully guess one of the primes is about $1/2^k$
% which is negligible.
% \item Second, another device who generates an RSA
% modulus using this protocol will choose the same prime with
% probability at most about $1/2^k$.  Technically this follows from the fact
% that the collision entropy of a distribution (denoted by $H_2$) is lower
% bounded by its min entropy (denoted by $H_\infty$).
% \end{itemize}
% These positive implications explain why it suffices to lower bound
% the min-entropy of each of the primes. 

\begin{theorem}  \label{thm:malice}
Consider an honest entropy authority interacting with a malicious
polynomial-time device using the protocol in
Figure~\ref{fig:proto-rsa}.  Let~$n$ be a modulus signed by an honest
entropy authority in Step~\ref{rsa:easign}.  
If $n=pq$, where  $p$ and
$q$ are primes in $[2^{k+1},2^{k+2})$, and the prime order $Q$ of the group 
used for Pedersen
commitments satisfies $Q > 2^{4k+8}$, then $n$ is sampled from a
distribution with at least $2k - c \log(k)$ bits of min-entropy for
some absolute constant $c$.
\end{theorem}

The proof of Theorem~\ref{thm:malice} relies on a number theoretic
statement about the number of small solutions to a particular modular
equation.  We first state the lemma and then prove Theorem~\ref{thm:malice}. 

\begin{lemma} \label{lemma:numth}
For all sufficiently large $k$, 
all primes $Q > 2^{4k+8}$, 
all $x,y \in \mathbb{Z}_Q$, and all $n \in [2^{2k+2}, 2^{2k+4})$ that are a 
product of two primes: \\
\mbox{}\quad the number of $x',y' \in [2^k, 2^{k+1})$ such that
\begin{equation*} 
  n \equiv (x+x')(y+y') \pmod{Q}  
\end{equation*}
is bounded by $k^d$ for some absolute constant $d$. 
\end{lemma}

\noindent
We prove Lemma~\ref{lemma:numth} in Appendix~\ref{app:proof}.

In our protocol, the device has control
over the values $\delta_x$ and $\delta_y$.
These values must fall in the range $[0, \Delta)$,
where $\Delta$ has size linear in $k$, so 
the values $x'$ and $y'$ in Lemma~\ref{lemma:numth}
range over $[2^k, 2^{k+1} + wk)$ for some absolute constant $w$.
We note that if Lemma~\ref{lemma:numth} holds over the interval
$[2^k, 2^{k+1})$, it must also hold over this wider interval
because the number of {\em additional} solutions in the wider interval
cannot be larger than $2wk$.
Thus, the total number of solutions in the
wider interval is still bounded by a polynomial in $k$.

\begin{proof}[Proof of Theorem~\ref{thm:malice}]
In Step~\ref{rsa:easign} of the protocol, the proof $\pi$ convinces the EA that
\[   n = (x+x'+\delta_x)(y+y'+\delta_y) \pmod{Q}  \]
for some (unknown) $x$ and $y$, where $Q$ is the group order used for 
the Pedersen commitments.  Recall that $Q>2^{4k+8}$.  
Since the device 
must commit to $x$ before seeing $x'$ we know that $p_0 = x+x'$ is sampled 
from a distribution over $\mathbb{Z}_Q$ with min-entropy at least $k$ (in the
worst case, $x+x'$ is sampled uniformly from the integers in an interval 
of width $2^k$).  Similarly the min-entropy of $q_0 = y+y'$ is at least $k$
even when conditioned on $p_0$.  Therefore, the probability that 
$(p_0, q_0)$ is equal to a particular pair in $\mathbb{Z}_Q$ is at 
most $1/2^{2k}$.

Now, by Lemma~\ref{lemma:numth}, for all RSA moduli $n$
in the interval $[2^{2k+2}, 2^{2k+4})$, the probability 
that $p_0 q_0 = n \bmod Q$ is at most $k^d / 2^{2k}$.
Since the density of RSA moduli in this interval is 
about $1/k^2$, conditioning on $n$ being an RSA modulus increases
the probability that $p_0 q_0 = n \bmod Q$ to at most $k^{d+2} / 2^{2k}$.
Therefore the min-entropy of the random variable $p_0 q_0$
conditioned on $p_0 q_0$ being an RSA modulus is at least
$2k - (d+2) \log k$.  
Since the malicious device controls $\delta_x$ and
$\delta_y$ and $0 \leq \delta_x, \delta_y < \Delta$ the device can increase
the probability of a particular $n$ by at most a factor of $\Delta^2$.
Therefore, assuming $\Delta < wk$ for some absolute constant $w$,
the probability that a successful protocol run produces a particular
RSA modulus is at most $k^{d+4}w^2 / 2^{2k}$.  We obtain that the min-entropy of 
the random variable $(x+x'+\delta_x)(y+y'+\delta_y) \pmod{Q}$, conditioned 
on this quantity being an RSA modulus, is at least $2k - (d+4) \log k - 2 \log w$,
as required. 
\end{proof}

We note that Lemma~\ref{lemma:numth} would not be needed to prove
Theorem~\ref{thm:malice} if in the message following Step~1
(Figure~\ref{fig:proto-rsa}) the device proved to the EA in
zero-knowledge that $x$ and $y$ are in the interval $[2^k,2^{k+1})$
using a zero-knowledge range proof~\cite{B00,CCS08}.  The reason is
that if $x$ and $y$ are bound to the relatively short interval
$[2^k,2^{k+1})$ then Lemma~\ref{lemma:numth} is trivial to prove.  By
relying on Lemma~\ref{lemma:numth}, which holds for all $x$ and $y$
in $\mathbb{Z}_Q$, we avoid the need for zero-knowledge range proofs,
making our protocol considerably more efficient.

\paragraph{Smaller Values of $Q$}
While we proved Lemma~\ref{lemma:numth} (and thus
Theorem~\ref{thm:malice}) for $Q > 2^{4k+8}$, we conjecture that
Theorem~\ref{thm:malice} holds for smaller values of $Q$ and in
particular when $Q > 2^{2k+100}$ as suggested in
Section~\ref{sec:RSAgen}.  This improves efficiency of the protocol
since Pedersen commitments are more efficient with a smaller $Q$.  

The reason that security likely holds for a smaller $Q$ is that that a
considerably weaker version of Lemma~\ref{lemma:numth} is sufficient
to prove the security of our protocol.  To see why, observe that the
$x,y \in \mathbb{Z}_Q$ chosen by device must be such that for random
$x'$ and $y'$ in $[2^k,2^{k+1})$ the resulting modulus $n$ is in the
correct range with non-negligible probability.  The set of such $x$
and $y$ is quite limited and it suffices that Lemma~\ref{lemma:numth}
hold only for such $x$ and $y$ which is a considerably smaller set
than all of $\mathbb{Z}_Q$.  
Clearly, this restriction on $x$ and $y$
is satisfied if $x$ and $y$ are in the range $[2^k,2^{k+1})$.
However, a malicious device can choose $x$ and $y$ differently. 
For example the device can choose $x$ in $[2^k,2^{k+1})$ and $y$ in
$[Q/2+2^k, Q/2+2^{k+1}]$.  In this case the resulting $n$ will be in
the correct range with probability $1/2$, namely whenever
$x+x'+\delta_x$ is even.  Since the device is restricted to choosing
such special $x,y \in \mathbb{Z}_Q$ (i.e. $x$ and $y$ such that for
random $x'$ and $y'$ in $[2^k,2^{k+1})$ the resulting modulus $n$ is
in $[2^{2k+2}, 2^{2k+4})$ with non-negligible probability), we only
need Lemma~\ref{lemma:numth} to hold for such $x$ and~$y$.

\paragraph{Optimization for Large $Q$}
An implementation that needs to use the provable bound on $Q$ from
Theorem~\ref{thm:malice} will need to use a group $G$ for the Pedersen
commitments and product proofs which has order $Q \approx 2^{4k+8}$,
compared with the $Q \approx 2^{2k+100}$ used in
Section~\ref{sec:RSAgen}.  To reduce the overhead of using this larger
group size, an implementation could use a group $G$ whose order $Q$ is a
product of two large primes each of size about $2^{2k+4}$.  This group
$G$ then is a direct product of two smaller groups $G_1$ and $G_2$
each of size about $2^{2k+4}$.   In the appendix, we prove
that Lemma~\ref{lemma:numth} holds for such composite $Q$ and therefore,
Theorem~\ref{thm:malice} continues to hold for such groups $G$.

When $G$ is a direct product of $G_1$ and $G_2$ an implementation can
uniquely describe an element $x \in G$ as a tuple $\langle x_1 \in
G_1, x_2 \in G_2 \rangle$.  The implementation can then perform the
commitments and product proofs over these reduced elements (once in
$G_1$ and once in $G_2$) in groups of order $m_1 \approx m_2 \approx
2^{2k+4}$.  By representing group elements in this way, the
implementation benefits from the provable bounds in Theorem~\ref{thm:malice},
but group exponentiations in this larger group take only {\em twice} as long as group
operations take in the smaller group we use in Section~\ref{sec:RSAgen}.
In contrast, the na\"{\i}ve method of using a group of {\em prime} order $Q \approx 2^{4k+8}$ would cause
group exponentiations to take {\em four times} as long in the larger
group than the same operations take in the smaller group.
Since a number of the protocol operations (e.g., finding the $\delta$ values) take time independent
of the size of the group used for commitments, 
we expect that performing commitments in this larger group would only 
cause a $1.5\times$ overall slowdown.

     %%%%  text for full version
}

\subsection{DSA Protocol}
\label{sec:sec:dsa}

In this section, we prove that the DSA key
generation protocol satisfies the 
security properties outlined in Section~\ref{sec:model}.

\subsubsection{Protects Device from a Malicious EA}
We first prove that a device with a strong entropy source leaks
no information about its secret key to the entropy
authority during a run of the protocol.

\begin{theorem}
If the device has a strong entropy source
(i.e., the device can sample repeatedly from the uniform
distribution over $\mathbb{Z}_Q$),
then the entropy authority can simulate its interaction
with the device.
\end{theorem}

\begin{proof}
We construct a simulator $S$ 
that, given a DSA public key $A = g^a$, outputs
the transcript 
$\langle C_x, x', A, \pi, \sigma \rangle$
of a protocol run between an honest device
and a malicious entropy authority 
$\langle \mathcal{A}_1, \mathcal{A}_2 \rangle$.
The simulator $S$ constructs the transcript
as follows:

\begin{minipage}{\columnwidth}
\begin{tabbing}
000 \= 000 \= \kill
$S(A)$:
\>\> set $C_x \xleftarrow{R} G$\\
\>\> set $x' \gets \mathcal{A}_1(C_x)$\\ 
\>\> generate $\pi$ using the NIZK simulator\\
\>\> set $\sigma \gets \mathcal{A}_2(A, \pi)$\\ 
\>\> if $\sigma$ is invalid, rewind the adversary and repeat;\\
\>\> otherwise, output the simulated transcript: \\
\>\> $\qquad\langle C_x, x', A, \pi, \sigma \rangle$
\end{tabbing}
\end{minipage}\\

This simulated transcript is indistinguishable from the
transcript an EA would generate during an interaction with an
honest device with a strong entropy source.
The value $C_x$ will be a random element from $G$ in both cases,
the value $x'$ will be chosen by the adversary in both cases,
the NIZK will be simulable in the random oracle model~\cite{bellare93random}
by the zero-knowledge property of the NIZK, and the signature is
constructed identically in either case.

As explained at the end of the proof of Theorem~\ref{thm:rsa-device}, 
we can assume that the entropy authority allows an honest execution of
the protocol to succeed with
non-negligible probability and thus, the simulator runs in 
time polynomial in the security parameter.
\end{proof}

\subsubsection{Protects Device from the CA and Client}
Having established that a device with a strong entropy source leaks
no secret information to the entropy authority, we
now demonstrate that the secret key produced by the protocol
is sampled from the uniform distribution over
the set of possible keys, even if the entropy authority is dishonest.

\begin{theorem}
The secret key produced by a successful run of the protocol in
Figure~\ref{fig:proto-dsa} between an honest device (with a strong
entropy source) and a malicious entropy authority, will be sampled independently 
from the uniform distribution over~$\mathbb{Z}_Q$.
\end{theorem}

\begin{proof}
Given that the device is honest and has a strong entropy source, the device
will sample the commitment randomness $r$ from the 
uniform distribution over $\mathbb{Z}_Q$. 
Thus, the commitment $C_x$ will be independent of the device's secret value $x$.
The entropy authority must send its random value $x'$ to the device given only 
this commitment $C_x$, so no matter how the entropy authority selects $x'$, it
must be independent of~$x$.
The honest device forms its secret key as $a = x+x' \bmod Q$, and since $x$ is
independent of $x'$, the secret key $a$ will be sampled from the uniform
distribution over $\mathbb{Z}_Q$.
\end{proof}

\subsubsection{Protects EA from a Malicious Device}
Finally, we show that an honest entropy authority with a strong entropy source
will only sign public keys whose corresponding
private keys are sampled from the uniform distribution over~$\mathbb{Z}_Q$. 

\begin{theorem}
The secret key produced by a successful run of the
protocol in Figure~\ref{fig:proto-dsa} between an honest entropy authority
(with a strong entropy source) and a malicious device, will be sampled
independently from the uniform distribution over~$\mathbb{Z}_Q$.
\end{theorem}

\begin{proof}
Since the device must commit to~$x$ before it sees~$x'$, $x$ must be
independent of~$x'$.  If the device could pick $x$ to depend on the authority's
value $x'$, the device would violate the binding property of the
commitment scheme.

The honest entropy authority will only sign the public key $A$ if the non-interactive
zero-knowledge proof $\pi$ the device sends in Step~\ref{dsa:pub} is valid and
the entropy authority will only accept the proof $\pi$ if $A = g^{x+x'}$.  
The hypothesis of the theorem is that the entropy authority is honest and has
a strong source of entropy, so the entropy authority will sample 
$x'$ uniformly from $\mathbb{Z}_Q$.
Since $x'$ is independent of $x$, the secret key $a=x+x' \bmod Q$ is uniform over
$\mathbb{Z}_Q$.
\end{proof}

\section{Evaluation}
\label{sec:eval}
\begin{table*}
\centering
\begin{tabular}{r | c c c c | c c c c}
& \multicolumn{4}{l | }{{\bf EC-DSA} (224-bit prime)} & \multicolumn{4}{l}{{\bf RSA} (2048-bit)} \\
    \hline
      & No proto & Proto & Proto+Net & {\em Slowdown} & No proto & Proto & Proto+Net & {\em Slowdown} \\
    \hline
Linksys Router  & 0.35 & 0.98 & 1.54    & $4.4\times$   & 50.54 & 93.78& 104.47& $2.1\times$\\
Laptop          & 0.014 & 0.085 & 0.646 & $48\times$  & 0.41 & 1.18 & 2.05 & $5.0\times$\\
Workstation     & 0.003 & 0.052 & 0.638 & $200\times$ & 0.15 & 0.65 & 1.22 & $8.3\times$\\
\end{tabular}
\caption{Time (in seconds) to generate a keypair
without our protocol, with a local EA,
and with an EA via the Internet with $\approx100$ ms of round-trip latency.
The Slowdown column indicates the rounded slowdown
factor of our protocol running over the Internet 
relative to the standard key generation algorithm.}
\label{tab:timing}
\end{table*}

To demonstrate the practicality of our RSA and DSA
key generation protocols,
we implemented the protocols in C using
the OpenSSL cryptography library.
We evaluated the performance of the
protocols on three different devices: 
a Linux workstation with two 3.2 GHz Intel W3656 processors,
a MacBook Pro laptop with a single 2.5 GHz dual-core processor, 
and a Linksys E2500-NP home router with a 300 MHz
Broadcom BCM5357r2 processor.
The entropy authority in all experiments 
was a modern Linux server and 
the DSA protocol experiments use the
NIST P-224 elliptic curve as the elliptic curve
DSA (EC-DSA) group~\cite{gallagher09fips}.
The source code of our implementation is available
online at \url{http://github.com/henrycg/earand}.

Embedded devices, like the Linksys router we used
in our evaluations, lack the keyboard, mouse, hard drive,
and other peripherals used as entropy sources on full-fledged machines.
As a result, these device are particularly susceptible
to generating weak keys.
By evaluating our key generation protocols on 
a \$70 Linksys router, we demonstrate
that the protocols are practical
even on low-power, low-cost (and often low-entropy) embedded devices.
For the purposes of evaluation, we installed the 
Linux-based dd-wrt~\cite{dd-wrt} operating system
on the Linksys router and ran our key generation protocol in a
user-space Linux process. 

Table~\ref{tab:timing} presents the
wall-clock time required to generate
a 2048-bit RSA key and a 224-bit EC-DSA key
on each machine, averaged over eight trial runs.
When running on the laptop and workstation, 
which have relatively fast CPUs, 
the bulk of the protocol overhead 
comes from the network latency in communicating with the
entropy authority.
On the CPU-limited home router,
the RSA protocol causes a near-$2\times$ slowdown.
Running the EC-DSA protocol 
takes fewer than two seconds on all three of the
devices.

The standard RSA keypair generation algorithm
requires much more computation than the EC-DSA
algorithm, so the cost of interacting with 
the entropy authority is amortized over
a longer total computation in the RSA protocol.
As a result, the slowdown factors on each of
the three devices is smaller for the RSA protocol
than for the DSA variant.
The protocol incurs a $2.1\times$
slowdown when running on the home router---%
generating a standard 2048-bit 
RSA keypair takes roughly 50 seconds and
generating a keypair with the protocol takes
just over 104 seconds.
On the laptop and workstation, around 50\% of the
slowdown is due to network latency.
On these faster devices, generating an
RSA keypair using the protocol takes less 
than three seconds.

% Ugly hack to put math symbols in bar char
\let\OrgDelta\delta
\protected\def\delta{%
\ifincsname
  \string\delta%
\else
  \OrgDelta
\fi
}

\let\OrgPound\#
\protected\def\#{%
\ifincsname
  \string\#%
\else
  \OrgPound
\fi
}

\begin{figure}
\begin{tikzpicture}
\begin{axis}[xlabel={CPU user time (seconds)},
xbar,
xmin=0,
enlarge y limits=0.2,
width=0.32\textwidth,
height=0.22\textwidth,
ytick=data,
bar width=7pt,
symbolic y coords={%
{Generate PKCS\#10 req.},
{Generate $C_x,C_y$},
{Multiplication NIZK},
{Generate $C_p, C_q$},
{Find $\delta_x,\delta_y$}
}]

\addplot [pattern=crosshatch] coordinates
{(1.4,{Generate PKCS\#10 req.})
(12.2,{Generate $C_x,C_y$})
(16.35,{Multiplication NIZK})
(20.1,{Generate $C_p, C_q$})
(39.5,{Find $\delta_x,\delta_y$})
};
\end{axis}
\end{tikzpicture}
\caption{Operations taking longer than 0.05s
during a run of the RSA protocol on the
home router.}
\label{fig:breakdown}
\end{figure}
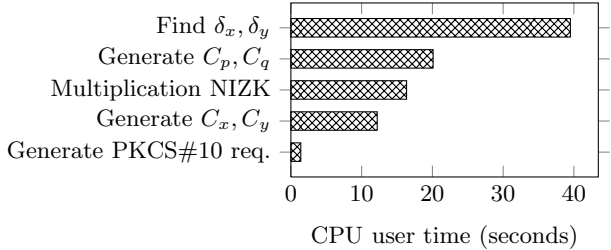

Figure~\ref{fig:breakdown} presents 
a graphical break-down of the CPU
user time required to perform
the most expensive operations in the RSA key generation
protocol on the home router.
Nearly half of the CPU time consumed during
the protocol is spent in finding the
$\delta_x$ and $\delta_y$ offset values
to make the RSA factors $p$ and $q$ prime.
Finding these offsets requires running 
the Miller-Rabin~\cite{rabin80probabilistic}
primality test on a number of candidate primes.
This expensive search for primes $p$ and $q$ 
is also required to generate an RSA modulus
{\em without} our key generation protocol,
so this search does not constitute protocol overhead.

The other expensive operations are computing
the Pedersen commitments (each of which requires
big-integer modular exponentiations) and 
generating the non-interactive zero-knowledge 
proof that $n$ is the product of the values
contained in the commitments $C_p$ and $C_q$.
The final expensive operation is 
generating the PKCS\#10 certificate 
request, which the device signs with 
its newly generated RSA key. 

Our implementation does not 
use fast multi-exponentiation 
algorithms~\cite{moller01algorithms}
(e.g., for computing Pedersen commitments $g^a h^r$ quickly)
or exploit parallelism to increase
performance on multi-core machines.
An aggressively optimized production-ready 
implementation could use these
techniques to improve the performance of the protocol.

As shown in Figure~\ref{fig:keysize}, 
our protocol imposes a near-uniform $6\times$
computational overhead (measured in CPU user time) 
on EC-DSA key generation.
This slowdown arises because our EC-DSA
protocol requires five elliptic curve point multiplications
and a single signature verification, compared with the
single elliptic curve point multiplication
required in traditional EC-DSA key generation.
At the smallest usable EC-DSA key size, 112 bits, the protocol set-up cost 
dominates the overall running time, so the protocol imposes 
a near-$9\times$ overhead.

The computational overhead of generating RSA keys using our protocol
decreases as the key size increases.
The dominant {\em additional} cost of our RSA protocol is the cost
of the modular exponentiations used in the commitment scheme
and zero-knowledge proof generation.
As $k$ increases, the cost of finding the RSA primes
grows faster than the additional cost of our protocol,
so the computational overhead of our protocol tends to 1.

\begin{figure}
\centering
%
% RSA 
%
\begin{tikzpicture}[domain=0:3]
\begin{axis}[scale only axis,
legend pos=north east,
width=0.39\textwidth,
ymin=0,
ymax=14,
height=1.4in,
xtick={512, 1024, 1536, 2048, 3072, 4096},
scatter/classes={
      0={mark=square*,blue, scale=1.5},%
      1={mark=triangle*,OliveGreen, scale=2},%
      2={mark=o,draw=orange, scale=2}},
    ytick={0, 2, 4, 6, 8, 10, 12, 14},
    yticklabels={$0$, $2\times$, $4\times$, $6\times$, $8\times$, $10\times$, $12\times$},
    xlabel={RSA key size (bits)},
    ylabel style={align=center},
    axis y line*=left,
    axis x line*=bottom,
    ylabel={Computational overhead}]
\addplot[scatter, scatter src=\thisrow{class},
        ultra thick,
        color=OliveGreen,
        error bars/.cd, 
        y dir=both, 
        x dir=none, 
        y explicit, 
        x explicit, 
        error bar style={color=OliveGreen, thick}]
              table[x=x,y=y,x error=xerr,y error=yerr] {
% 12.72670807 6.255813953 3.405185409 2.278610624 1.595497108 1.958379602 
x       xerr    y             yerr        class
512     0.000   12.72670807   0.0         1
1024    0.000   6.255813953   0.0         1
1536    0.000   3.405185409   0.0         1
2048    0.000   2.278610624   0.0         1
3072    0.000   1.595497108   0.0         1
4096    0.000   1.958379602   0.0         1
            };

\legend{EC-DSA, RSA}
\end{axis}
%
% EC-DSA
%
\begin{axis}[
scale only axis,
width=0.39\textwidth,
height=1.4in,
ymin=0,
ymax=14,
scatter/classes={
      0={mark=square*,blue, scale=1.5},%
      1={mark=triangle*,OliveGreen, scale=2},%
      2={mark=o,draw=orange, scale=2}},
ytick=\empty,
xlabel={EC-DSA key size (bits)},
xlabel near ticks,
xtick={112, 160, 192, 224, 256, 384, 521},
xticklabels={112, , 192, , 256, 384, 521},
domain=0:3,
axis y line*=none,
axis x line*=top,
axis y line*=right,
axis x line*=top]
\addplot[scatter, scatter src=\thisrow{class},
        color=blue,
        ultra thick,
        error bars/.cd, y dir=both, x dir=both, y explicit, x explicit, 
        error bar style={color=blue, thick}]
        table[x=x,y=y,x error=xerr,y error=yerr] {
% 8.959322034 7.897832817 7.993975904 6.497607656 7.95184136  6.057591623 6.764150943
x       xerr    y           yerr    class
112     0.000   8.959322034 0.0     0
160     0.000   7.897832817 0.0     0
192     0.000   7.993975904 0.0     0
224     0.000   6.497607656 0.0     0 
256     0.000   7.95184136  0.0     0 
384     0.000   6.057591623 0.0     0
521     0.000   6.764150943 0.0     0
              };
\end{axis}
\end{tikzpicture}
 
\caption{Computational overhead 
(in CPU user time) imposed when a
laptop uses our key generation protocols 
to generate keypairs of various sizes,
averaged over 32 trials.}
\label{fig:keysize}
\end{figure}
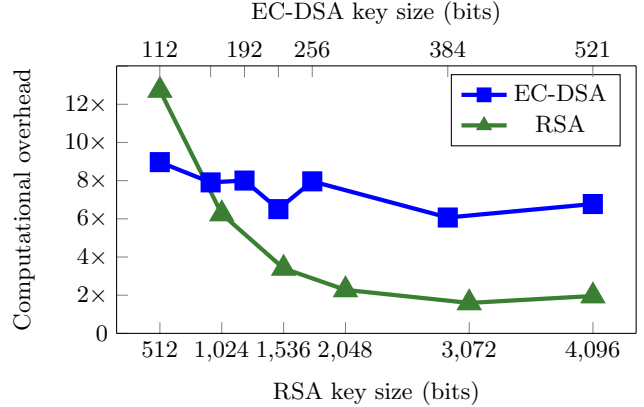

\pagebreak[4]
\section{Implementation Concerns}
\label{sec:impl}

This section discusses a handful of practical implementation issues
that a real-world deployment of our key generation protocols would have
to address. 

\paragraph{Integration with the CA infrastructure}
Integrating our key generation protocols with
the existing CA infrastructure would require 
only modest modifications to today's infrastructure.
In a deployment of our key generation protocol,
the device could interact with the entropy authority using
an HTTP API. 
After the device obtains the entropy authority's signature
on its public key, the device would embed the EA signature
in an extension field in the PKCS\#10 certificate signing
request that the device sends to the certificate authority.
Each certificate authority would maintain a list of public
keys of approved entropy authorities 
(in the way that browsers and SSL libraries today maintain a list
of root CA public keys).
When a certificate authority receives a PKCS\#10 
request from a device, the CA would first check the validity of
the EA's signature on the request. 
If the signature is valid and the CA is able to verify
the identity of the requesting device, the CA would sign the certificate
and return it to the device.

We expect that many commercial certificate authorities 
would be willing to serve as free public entropy authorities,
since the computational cost of acting as an 
entropy authority is small (less than one CPU-second per protocol run).
Organizations large enough to have their own IT departments might
run their own internal entropy authorities as well.

\paragraph{Self-Signed Certificates and SSH}
TLS servers often use self-signed
certificates to provide link encryption without
CA-certified identity.
The analogue of a self-signed certificate in our
setting is a certificate that is signed by the
entropy authority but that is {\em not} signed
by a certificate authority.
This sort of certificate would convince a third party
that the device's public key is sampled from a high-entropy 
distribution, without convincing a third party that the key
corresponds to a particular real-world identity.
As long as some EAs provide their services for free 
(which we expect), EA-signed certificates
will be free, just as self-signed certificates are free today.

To generate such a certificate, the device
would submit a PKCS\#10 certificate signing request to the entropy 
authority at the end of Step~\ref{rsa:delta}
of the RSA protocol or Step~\ref{dsa:pub} of
the DSA protocol, along with other data it sends.
The entropy authority would then sign the request
and would return the EA-signed certificate to the device.
TLS clients (e.g., Web browsers) would maintain a list
of public keys of approved entropy authorities, just as
today's client keep a list of approved root CAs.
When a client connects to a device that uses an EA-signed
certificate, the client would verify the EA's signature and 
would treat the certificate just as 
it treats self-signed certificates today.

SSH could similarly use EA-signed 
keys to use convince clients that
a particular SSH host generated 
its public key using random values
from an approved entropy authority.
To accomplish this, the SSH server software 
would define a new public key algorithm type
for EA-signed keys (e.g., {\tt ssh-rsa-rand}).
Keys of this type would contain the
SSH host's normal public key, but they would also
contain an EA's signature on the SSH host's public key
(along with the fingerprint of the signing EA's key).
SSH clients that support the {\tt ssh-rsa-rand} key type
would be able to verify the EA's signature on the 
host's key to confirm that that the host's
key incorporates randomness from an approved
entropy authority.

\paragraph{Other entropy issues}
Our key generation protocol {\em only} ensures that
a device's RSA or DSA keypair has sufficient randomness---it does
not ensure randomness in other security-critical parts of the system
(e.g., signing nonce generation, TLS session key selection,
address space layout randomization).
We focus on cryptographic key generation because
attacks against weak public keys are especially easy to mount.
Once a device publishes a weak public key,
the device is likely to use the same public key for months or years.
Thus, even if the device's entropy source strengthens
over time (as the device gathers randomness from
network interrupt timings, for example) 
the device's keys remains weak.
Hedged public-key cryptography~\cite{bellare09hedged,ristenpart10good},
in conjunction with our key generation protocols, would help 
reduce the risk of bad randomness in signing and encryption,
but solving all of these randomness problems is likely beyond
the scope of any single system.

\paragraph{Distributing trust with many entropy authorities}
As we note in Section~\ref{sec:model}, if the device has a 
weak entropy source then there is no way to protect the device against
an eavesdropper that observes {\em all communication} between
the device and the EA.
Our threat model excludes the possibility of such an eavesdropper,
but if the device is particularly concerned about eavesdroppers on
its initial conversation with the EA, the device
could run a modified version of the protocol with {\em many} 
entropy authorities instead of just one.
With multiple EAs, an eavesdropper 
would have to observe the device's communication 
with {\em all} of the EAs
to learn the device's secret key.
Informally, if an adversary controls all but one of the entropy authorities
(call this entropy authority the ``honest'' one) and if the adversary can
eavesdrop on the device's communication with all entropy authorities {\em
except} the honest one, then a private key $a$ generated using
a multi-authority key-generation protocol will still be sampled from
the uniform distribution, even when conditioned on the adversary's knowledge.

We sketch the multi-authority DSA protocol here.
A similar modification allows RSA key generation 
with multiple entropy authorities.
In the following protocol, each entity has a well-known
long-term signature verification public key and every message sent
between participants includes a public per-session nonce and
is signed with the sender's long-term public key.
The multi-authority DSA protocol proceeds as follows:
\begin{itemize}
  \item The device commits to its random $x$ using
        randomness $r$ and sends $C \gets \textsf{Commit}(x;r)$
        to each of the $N$ entropy authorities.
  \item Each entropy authority $i$ 
        selects random values $x_i$ and $r_i$,
        generates a commitment $C_i \gets \textsf{Commit}(x_i; r_i)$,
        produces a signature $\sigma_i$ on 
        $\langle i, C, C_i \rangle$
        and returns $\langle x_i, r_i, \sigma_i \rangle$ to the device.
  \item The device sets $A = g^{x+\Sigma_i x_i}$ and produces a 
        non-interactive zero-knowledge
        proof of knowledge:
        \begin{align*}
        \textsf{PoK}\{&x, x_1, \dots, x_N, r_1, \dots, r_N:\\
          &C = g^x h^r 
            \land (\land_i C_i = g^{x_i} h^{r_i}) 
            \land A = g^{x + \Sigma_i x_i}
        \}
        \end{align*}
        incorporating the randomness provided by the entropy authorities
        into the secrets it uses to generate the non-interactive proof.
        The device sends 
        \[\langle A, \textsf{PoK}, C_1, \dots, C_N, \sigma_1, \dots, \sigma_N \rangle\]
        to each entropy authority.
  \item Each entropy authority $i$ checks each signature $\sigma_j$,
        verifies the zero-knowledge proof, signs the device's public key $A$,
        and returns its signature on the public key to the device.
\end{itemize}

Without loss of generality, assume that entropy authority 1 is the honest one. 
Since we assume that the adversary cannot eavesdrop on the device's
communication with entropy authority 1, the adversary will never learn $x_1$
and it will only learn $C_1$ after it has
had to pick $C_2, \dots, C_N$. 
Thus, the adversary's commitments must be independent of $C_1$ and
its values $x_2, \dots, x_N$ must be independent of $x_1$.
Since the device generates the proof of knowledge using randomness derived from
the randomness sent by each entropy authority (including entropy authority 1),
the proof of knowledge will not leak any information about the device's secrets. 
The device's private key $a = x + \Sigma_i x_i \bmod Q$ will then be sampled from
the uniform distribution over $\mathbb{Z}_Q$, even conditioned on the adversary's
knowledge.

\paragraph{Default keys}
Roughly 5\% of TLS hosts on the Internet in 2012
used {\em default keys}, which are pre-loaded into the
device's firmware by the manufacturer~\cite{heninger12mining}.
Typically, any two such devices of the same model and firmware version will 
ship with the same public and secret key.
To recover a default secret key, an attacker can download the firmware
for the device from the manufacturer's Web site or look up
the default key in a database designed for that purpose~\cite{littleblackbox}.

Our protocol does not protect against a manufacturer who 
installs the same keypair in many devices.
If a manufacturer wants all of her devices to ship with a default
keypair signed by an entropy authority, the manufacturer could run
our key generation protocol {\em once} in the factory, and then
install this single EA-signed keypair in every device shipped.

Installing the same keypair in many devices is tantamount to
publishing the device's secret key, which is an ``attack'' 
which we cannot hope to prevent.
As a heuristic defense against default keys, a client connecting
to a device could require that the device use a certificate
that was generated after the manufacture of the device
(as indicated, for example, by an EA-signed timestamp on the certificate).

\section{Related Work}
\label{sec:rel}

Hedged public-key cryptography~\cite{bellare09hedged,ristenpart10good}
addresses the problem of weak randomness during message {\em signing} or {\em encryption},
whereas our work addresses the problem of randomness during {\em key generation}.
Cryptographic hedging provides no protection against randomness failures
when generating cryptographic keys but deployed systems could use 
hedging in conjunction with our
key-generation protocols to defend against weak randomness after generating
their cryptographic keys.

Intel's Ivy Bridge processor implements a hardware instruction 
that exploits physical uncertainty in a dedicated circuit
to gather random numbers~\cite{taylor11behind}.
A hardware random number generator provides a new and potentially
rich source of entropy to cryptographic applications.
Devices without hardware random number generators 
could use a variety of other techniques to gather
possibly unpredictable values early in the system boot 
process~\cite{mowery13welcome}.
Even so, having a rich entropy source does not mean that software
developers will properly incorporate the entropy into cryptographic secrets.
Our protocol ensures that keys will have high entropy, even if the
cryptographic software ignores or misuses hardware-supplied randomness.

\begin{figure}
\begin{tabular}{l | c c}
& Single exp.& Double exp. \\
\hline
Juels-Guajardo protocol~\cite{juels02verifiable} & 319 & 35 \\
This paper & 8 & 4 \\
\end{tabular}
\caption{Approximate number of 
  $k$-bit modular exponentiations the
  device must compute to generate
  a $k$-bit RSA modulus.}
\label{fig:juels}
\end{figure}

Juels and Guajardo~\cite{juels02verifiable} offer a protocol for
RSA key generation that is superficially similar
to the one we present here.
The Juels-Guajardo protocol protects against 
{\em kleptography}~\cite{young97kleptography}, in
which a device's cryptography library is adversarial, and
{\em repudiation}, in which a signer intentionally generates a weak
cryptographic signing key so that the signer can disown signed
messages in the future.
To prevent against these very strong adversaries, 
their protocol requires a number of 
additional zero-knowledge proofs 
that are unnecessary in our model.
Using the number of modular exponentiations as a
proxy for protocol execution time, the Juels-Guajardo
protocol would likely take over 40 minutes to execute
on the home router we used in our experiments, while
our protocol takes fewer than two minutes
(see Figure~\ref{fig:juels}).
In addition, Juels and Guajardo do {\em not}
address the issue of a device whose source of randomness
is so weak that it cannot create blinding commitments or
establish a secure SSL session.

\section{Conclusion}
This paper presents a systemic solution to the 
problem of low-entropy keys.
We present a new threat model, in which a device
generating cryptographic secrets may have one
communication session with an {\em entropy authority} 
which an eavesdropper cannot observe.
Under this threat model, we describe protocols
for generating RSA and DSA keypairs that 
do not weaken keys for devices that have a strong
entropy source, but that can considerably
strengthen keys generated on low-entropy devices.
Our key generation protocols incur tolerable slow-downs,
even on a CPU-limited home router.
The threat model and protocols presented herein
offer a promising solution to the long-standing
problem of weak cryptographic keys.

\subsection*{Acknowledgements}
We gratefully acknowledge Justin Holmgren for pointing out an error in the DSA
security proof of the proceedings version of this paper.  
We thank David
Wolinsky, Ewa Syta, and Zooko Wilcox-O'Hearn for their helpful comments.  
This material is based upon work supported by NSF, the Defense Advanced
Research Agency (DARPA) and SPAWAR Systems Center Pacific, Contracts No.
N66001-11-C-4018 and N66001-11-C-4022, and the National Science Foundation
Graduate Research Fellowship under Grant No. DGE-114747.  
This work was also
supported by a Google faculty award.

\bibliographystyle{abbrv}
\bibliography{papers}

\abbr{}{
\appendix
\section{Proof of Lemma~\reflemmanumth}
\label{app:proof}

We prove Lemma~\ref{lemma:numth} using a technique from Cilleruelo and 
Garaev~\cite[Proof of Theorem 1]{cilleruelo2011concentration}.
First we state the following simple fact.  In what follows, the quantity
$x \bmod m$ refers to an integer $|y|<m/2$ such that $x = y \pmod{m}$. 
We will use this fact where the scalar $c$ (below) is some small constant
such as $c=1$ or $c=2$.

\begin{fact} \label{easyfact}
For all integers $m>2,\ c>0$, and $a,b \in \mathbb{Z}$ 
there exists a non-zero integer $t$ such that $|t| < m^{1/2}/c^2$ and
\[  |a t \bmod m| < c m^{3/4}, \qquad   
                           |b t \bmod m| < c m^{3/4}  \ .   \]
\end{fact}

\begin{proof}
Consider the integer lattice spanned by three vectors 
$\mathbf{u}_1 = (c^3 m^{1/4}, a, b)$,\ \ \
$\mathbf{u}_2 = (0,m,0)$ and $\mathbf{u}_3 = (0,0,m)$.  
Its determinant is $c^3 m^{2.25}$ and therefore, by Minkowski's theorem,
this lattice must contain a non-zero point $\mathbf{u}$ whose $\ell_\infty$
norm is less than $c m^{3/4}$.  The required integer $t$ is the
multiplier of $\mathbf{u}_1$ in the integer linear combination used to
form $\mathbf{u}$ from the three basis vectors.   This $t$ must be non-zero
because $\mathbf{u}$ is non-zero.
\end{proof}

\bigskip
Next we prove Lemma~\ref{lemma:numth}.  We rename the variables
for clarity and allow the modulus $m$ (called $Q$ in Lemma~\ref{lemma:numth})
to be either prime or a composite with no small prime factors.
We slightly generalize the lemma and note that in Lemma~\ref{lemma:numth}
we fixed $c=1$, but here we allow $c$ to be arbitrary.  The optimization
at the end of Section~\ref{sec:security} will use $c=2$ and a modulus $m$
whose prime factors are all greater than $\sqrt{m}/4$.

\begin{lemma*}
For $c>0$ let $m$ be a positive integer whose smallest prime factor
is greater than $\sqrt{m}/c^2$.  Then for all 
constants $a,b \in \mathbb{Z}_m$, and all 
non-zero $|n| < m^{1/2}$ such that $n$ is a product of at most two primes: \\
\mbox{}\quad the number of solutions $x,y$ to the equation
\begin{equation} \label{eq:lemma}
  n = (x+a)(y+b) \pmod{m}  
\end{equation}
with $|x|,|y| < m^{1/4}/2$ is at most polynomial in $c \log m$.
\end{lemma*}

\begin{proof}
First, we can assume that one of $a$ or $b$ is non-zero since
otherwise~{(\ref{eq:lemma})} clearly only has a constant number of
small solutions.  From Fact~\ref{easyfact} we find a non-zero integer
$|t|<m^{1/2}/c^2$ such that the integers 
$u = t a \bmod m$ and $v = t b \bmod m$  are both less than $c m^{3/4}$ in 
absolute value.  By assumption on the prime factors of $m$ we know 
that $\gcd(t,m)=1$.  It follows that one of $u$ or $v$ must be
non-zero, since otherwise $a=b=0$.

Without loss of generality we can assume that $\gcd(t,u,v)=1$ since
otherwise we can replace $t,u,v$ by their values divided by the gcd
and ensure that $\gcd(t,u,v)=1$.  This only reduces the magnitude of
$t,u,$ and $v$ and therefore the bounds on $|t|,|u|,|v|$ continue to
hold. Moreover, since $\gcd(t,m)=1$, after division by $\gcd(t,u,v)$
it still holds that $u = t a \bmod m$ and $v = t b \bmod m$.

Multiplying~{(\ref{eq:lemma})} by $t$, we obtain
\begin{equation}\label{eq:lemma-equiv}
txy + vx + uy \equiv tn - tab \pmod m 
\end{equation}
The left-hand side of this congruence satisfies:
\[ |txy + vx + uy| \leq m^{1/2} m^{1/4} m^{1/4} + c m^{3/4} m^{1/4} = (c + 1) m \]
Therefore, setting $r_0 = tn - tab \bmod m$ where $|r_0|<m/2$,
we can rewrite the congruence (\ref{eq:lemma-equiv}) as an equation:
\begin{equation}\label{eq:copp}
txy + vx + uy = r_0 + zm \qquad \text{for some integer $|z| \leq c+1$}
\end{equation}
There are now two cases to consider:
\begin{itemize}
\item\textbf{Case (i)}:
$|t| > m^{1/4}/2$ or $|u| > m^{1/2}/4$ or $|v| > m^{1/2}/4$.
In this case, for each integer $|z| \leq c+1$ we can apply Coppersmith's 
algorithm~\cite{coppersmith1997small} to~{(\ref{eq:copp})} to find small roots
$x$ and $y$ of the polynomial. 
To apply Coppersmith's algorithm to find the roots of a polynomial, 
the coefficients of the polynomial must share no common factor
and this holds since $\gcd(t,u,v)=1$. 

Coppersmith's algorithm finds roots 
$x \leq X$ and $y \leq Y$ of a polynomial $f(x,y)$ 
with degree $\delta$ in each variable separately,
when the largest coefficient $W$ of $f(Xx, Yy)$ 
satisfies: $XY < W^{2/(3\delta)}$.
For our purposes, we set $\delta=1,\ X = m^{1/4}/2$, and
$Y = m^{1/4}/2$ which are the bounds on $|x|$ and $|y|$ in
the theorem statement.

The value $W$ then is at least:
\begin{align*}
\max\{tXY, vX, uY\} &< W\\
\max\{tm^{1/2}/4, vm^{1/4}/2, um^{1/4}/2\} &< W
\end{align*}
In this case we have that either $|t| > m^{1/4}/2$,
$|u| > m^{1/2}/4$ or $|v| > m^{1/2}/4$ and therefore $W > m^{3/4}/8$.
Thus, we can apply Coppersmith's algorithm, since
\[XY = m^{1/2}/4 = (m^{3/4}/8)^{2/3} < W^{2/3} \ .  \]

Now, for each $|z| \leq c+1$, since Coppersmith's algorithm finds all 
small solutions to (\ref{eq:copp}) in time polynomial in $\log W$, 
and $W < m$, there can be at most
$\polylog(m)$ solutions to (\ref{eq:copp}) as required.
Since $|z|<c+1$, overall there are at most $(c+1) \polylog(m)$ solutions.

\item\textbf{Case (ii)}: 
$|t| \leq  m^{1/4}/2$ and $|u| \leq m^{1/2}/4$ and $|v| \leq m^{1/2}/4$.
Since $m^{1/4}/2$ cannot be an integer we know that $|t|$ is strictly
less than $m^{1/4}/2$.  Multiplying (\ref{eq:lemma-equiv}) by $t$, we can write
\[ (tx + u)(ty + v) \equiv t^2 n \pmod m \]
In this case, both the left-hand and right-hand sides of this equation
are strictly less than $m/2$ in absolute value.   Therefore, the equation
\begin{equation}\label{eq:lemma-int}
(tx + u)(ty + v) = t^2 n 
\end{equation}
must hold over the integers.

Since by assumption one of $u$ or $v$ is non-zero, let us suppose
that $u$ is the non-zero one (the case that $v$ is non-zero is handled 
identically).
In this case we can assume that $\gcd(t, u) = 1$, since
otherwise we can divide both sides of (\ref{eq:lemma-int}) by
$\gcd(t,u)$.  Since $\gcd(t, u) = 1$ we obtain that
$$\gcd(tx+u, t) = \gcd(t, u) = 1 $$
and consequently the only way that (\ref{eq:lemma-int}) can hold
is if $tx+u$ divides $n$.  Since $n$ is the product of at most two primes, 
$tx+u$ can only be $\pm 1, \pm n$, or one of the prime factors 
of $n$ (or their negatives).  Overall there can be at most eight 
solutions to~{(\ref{eq:lemma-int})} and therefore at most eight 
solutions to~{(\ref{eq:lemma})}.
\end{itemize}

\noindent
Either way there can be at most $\polylog(m)$ solutions to~{(\ref{eq:copp})}
as required.
\end{proof}

}

\end{document}